\def\extflag{}

\documentclass[acmsmall,dvipsnames,screen]{acmart}

\PassOptionsToPackage{table}{xcolor}

\setcopyright{cc}
\setcctype{by}
\acmJournal{PACMPL}
\acmYear{2026} \acmVolume{10} \acmNumber{OOPSLA1} \acmArticle{134}
\acmMonth{4} 
\acmDOI{10.1145/3798242}

\usepackage{booktabs}   
\usepackage{subcaption} 
\usepackage{amsfonts}
\usepackage{amsmath}

\usepackage{paralist}
\usepackage{listing}
\usepackage[ruled,algosection,noend]{algorithm2e}
\usepackage{multirow}
\usepackage{varwidth}
\usepackage{graphicx}
\usepackage{listings}
\usepackage{float}
\usepackage{multirow}
\usepackage[noend]{algorithmic}
\usepackage{mathpartir}
\usepackage{url}
\usepackage{hyperref}
\usepackage[capitalise]{cleveref}
\usepackage{xspace}
\usepackage{verbatim}
\usepackage{lipsum}
\usepackage{wrapfig}
\usepackage{ifsym}

\usepackage[inline, shortlabels]{enumitem}
\usepackage{url}
\usepackage{fancyvrb}
\usepackage[figurename=Figure]{caption}
\usepackage{nanoc}
\usepackage[T1]{fontenc}
\usepackage{relsize}
\usepackage{marvosym}
\usepackage{microtype}
\usepackage{adjustbox}
\usepackage{multicol}
\usepackage{tikz}
\usepackage{tcolorbox}
\usepackage{ragged2e}
\usepackage{pgfplots}
\usepackage{pgfplotstable}
\usepackage{arydshln}
\usepackage{array}
\newcolumntype{M}{>{\columncolor{gray!12}}c}

\usepackage{mathtools}
\usepackage{mathdots}
\usepackage{thmtools}
\usepackage[bottom,hang,flushmargin]{footmisc}
\usepackage{siunitx}
\usepackage{setspace}
\usetikzlibrary{arrows}
\usetikzlibrary{shapes.arrows}
\usetikzlibrary{shapes}
\usetikzlibrary{calc}
\usetikzlibrary{positioning}
\usetikzlibrary{fit}
\usetikzlibrary{tikzmark}
\usetikzlibrary{snakes}
\usetikzlibrary{decorations.pathreplacing,calc}

\usepgfplotslibrary{groupplots}

\usepackage{tikzsymbols} 
\usepackage{tikz-qtree}

\tikzstyle{arrow}+=[thick,rounded corners=0.5em]
\tikzstyle{every picture}+=[remember picture,baseline]

\usepackage{multirow}
\usepackage{hhline}

\hypersetup{colorlinks,
  linkcolor=ACMDarkBlue,
  citecolor=ACMPurple,
  urlcolor=ACMDarkBlue,
  filecolor=ACMDarkBlue}

\floatname{algorithm}{Algorithm}

\SetAlFnt{\small}
\SetAlCapFnt{\small}
\SetAlCapNameFnt{\small}
\SetAlCapHSkip{0pt}
\IncMargin{-\parindent}

\floatstyle{plaintop}
\restylefloat{table}

\defcitealias{Coq-manual}{Coq}

\makeatletter 
\def\arcr{\@arraycr}
\makeatother

\floatname{algorithm}{Algorithm}

\definecolor{shadecolor}{gray}{1.00}
\definecolor{ddarkgray}{gray}{0.75}
\definecolor{darkgray}{gray}{0.30}
\definecolor{lightgray}{gray}{0.98}
\definecolor{splashedwhite}{rgb}{1.0, 0.99, 1.0}
\definecolor{whitesmoke}{rgb}{0.96, 0.96, 0.96}
\definecolor{mintcream}{rgb}{0.96, 1.0, 0.98}
\definecolor{ghostwhite}{rgb}{0.97, 0.97, 1.0}
\definecolor{floralwhite}{rgb}{1.0, 0.98, 0.94}
\definecolor{azuremist}{rgb}{0.94, 1.0, 1.0}
\definecolor{aliceblue}{rgb}{0.94, 0.97, 1.0}

\definecolor{beige}{rgb}{0.96, 0.96, 0.86}
\definecolor{bananamania}{rgb}{0.98, 0.91, 0.71}
\definecolor{unmellowyellow}{rgb}{1.0, 1.0, 0.4}
\definecolor{pastelyellow}{rgb}{0.99, 0.99, 0.59}
\definecolor{cosmiclatte}{rgb}{1.0, 0.97, 0.91}
\definecolor{aliceblue}{rgb}{0.94, 0.97, 1.0}
\definecolor{coolblack}{rgb}{0.0, 0.18, 0.39}
\definecolor{darkblue}{rgb}{0.0, 0.0, 0.55}
\definecolor{darkpowderblue}{rgb}{0.0, 0.2, 0.6}
\definecolor{denim}{rgb}{0.08, 0.38, 0.74}
\definecolor{hanpurple}{rgb}{0.32, 0.09, 0.98}
\definecolor{ballblue}{rgb}{0.13, 0.67, 0.8}
\definecolor{bondiblue}{rgb}{0.0, 0.58, 0.71}
\definecolor{caribbeangreen}{rgb}{0.0, 0.8, 0.6}
\definecolor{midnightblue}{rgb}{0.1, 0.1, 0.44}
\definecolor{oxfordblue}{rgb}{0.0, 0.13, 0.28}
\definecolor{payne}{rgb}{0.25, 0.25, 0.28}
\definecolor{bluegray}{rgb}{0.4, 0.6, 0.8}
\definecolor{darkelectricblue}{rgb}{0.33, 0.41, 0.47}
\definecolor{steelblue}{rgb}{0.27, 0.51, 0.71}
\definecolor{prussianblue}{rgb}{0.0, 0.19, 0.33}

\definecolor[named]{ACMBlue}{cmyk}{1,0.1,0,0.1}
\definecolor[named]{ACMYellow}{cmyk}{0,0.16,1,0}
\definecolor[named]{ACMOrange}{cmyk}{0,0.42,1,0.01}
\definecolor[named]{ACMRed}{cmyk}{0,0.90,0.86,0}
\definecolor[named]{ACMLightBlue}{cmyk}{0.49,0.01,0,0}
\definecolor[named]{ACMGreen}{cmyk}{0.20,0,1,0.19}
\definecolor[named]{ACMPurple}{cmyk}{0.55,1,0,0.15}
\definecolor[named]{ACMDarkBlue}{cmyk}{1,0.58,0,0.21}
\definecolor[named]{PaleGreen}{RGB}{196, 255, 231}
\definecolor[named]{PaleOrange}{RGB}{255, 213, 169}
\definecolor{intnull}{RGB}{213,229,255}

\newcommand{\capcal}[1]{$\mathcal{#1}$}

\newcommand{\etc}{\emph{etc}\xspace}
\newcommand{\ie}{\emph{i.e.}\xspace}

\newcommand{\eg}{\emph{e.g.}\xspace}

\newcommand{\etal}{\emph{et~al.}\xspace}

\newcommand{\aka}{\textit{a.k.a.}\xspace}

\newcommand{\wrt}{\emph{wrt.}\xspace}

\newcommand{\eqdef}{\triangleq}

\newcommand{\nonts}{\mathsf{S}\xspace}

\newcommand{\tname}[1]{\textsf{#1}\xspace}
\newcommand{\token}[1]{{\small\texttt{#1}}}
\newcommand{\code}[1]{\lstinline{#1}}

\newcommand{\AmbPA}[3]{%
  {#1}_{\textcolor{gray}{\,{#2},{#3}}}
}

\definecolor{pblue}{rgb}{0.13,0.13,1}
\definecolor{pgreen}{rgb}{0,0.5,0}
\definecolor{pred}{rgb}{0.9,0,0}
\definecolor{pgrey}{rgb}{0.46,0.45,0.48}

\definecolor{ckeyword}{HTML}{7F0055}
\definecolor{ccomment}{HTML}{3F7F5F}
\definecolor{cnumber}{HTML}{2A0099}

\makeatletter
\newenvironment{btHighlight}[1][]
{\begingroup\tikzset{bt@Highlight@par/.style={#1}}\begin{lrbox}{\@tempboxa}}
{\end{lrbox}\bt@HL@box[bt@Highlight@par]{\@tempboxa}\endgroup}

\newcommand\btHL[1][]{%
  \begin{btHighlight}[#1]\bgroup\aftergroup\bt@HL@endenv%
}
\def\bt@HL@endenv{%
  \end{btHighlight}%
  \egroup
}
\newcommand{\bt@HL@box}[2][]{%
  \tikz[#1]{%
    \pgfpathrectangle{\pgfpoint{1pt}{0pt}}{\pgfpoint{\wd #2}{\ht #2}}%
    \pgfusepath{use as bounding box}%
    \node[
      anchor=base west, 
      outer sep=0pt,
      inner xsep=2pt, 
      inner ysep=0pt, 
      rounded corners=2pt, 
      minimum height=\ht\strutbox,
      fill=black!18,
      line width=0.25pt,
      #1
    ]{\raisebox{0.5pt}{\strut}\strut\usebox{#2}};
  }%
}
\makeatother

\tikzset{
  hlReach/.style={fill=black!10, line width=0.25pt}, 
  hlDup/.style={fill=black!20, line width=0.25pt}, 
  hlBox/.style={draw=black, line width=1pt},
  hlBoxSolid/.style={hlBox, dash pattern=},
  hlBoxDashed/.style={hlBox, dashed},
  hlBoxDotted/.style={hlBox, dotted},
}

\lstdefinestyle{HL}{
    basicstyle=\scriptsize\ttfamily\linespread{0.9}, 
    columns=flexible,
    moredelim=**[is][{\btHL[hlDup]}]{`}{`},
    moredelim=**[is][{\btHL[hlReach]}]{!}{!},
    moredelim=**[is][{\btHL[fill=none,draw=black,line width=0.4pt,dashed]}]{@}{@},
}

\newcommand{\greta}{\tname{Greta}}
\newcommand{\tool}{\greta}
\newcommand{\yacc}{\tname{yacc}}
\newcommand{\beaver}{\tname{Beaver}}
\newcommand{\menhir}{\tname{Menhir}}

\SetCommentSty{mycommfont}

\newcommand\Bstrut{\rule[-1.1ex]{0pt}{0pt}}

\makeatletter
\protected\def\ccell#1#{%
  \kern-\fboxsep
  \@ccell{#1}%
}
\def\@ccell#1#2#3{%
  \colorbox#1{#2}{#3}%
  \kern-\fboxsep
}
\makeatother

\DeclareSymbolFont{symbolsC}{U}{txsyc}{m}{n}
\DeclareMathSymbol{\synth}{\mathrel}{symbolsC}{123}
\definecolor{bittersweet}{rgb}{1.0, 0.44, 0.37}
\definecolor{awesome}{rgb}{1.0, 0.13, 0.32}
\definecolor{blush}{rgb}{0.87, 0.36, 0.51}
\definecolor{brilliantlavender}{rgb}{0.96, 0.73, 1.0}
\definecolor{brightube}{rgb}{0.82, 0.62, 0.91}
\definecolor{brightlavender}{rgb}{0.75, 0.58, 0.89}
\definecolor{brinkpink}{rgb}{0.98, 0.38, 0.5}
\definecolor{notpink}{rgb}{0.8, 0.3, 0.3}

\newcommand{\pts}{\mapsto}

\newcommand{\set}[1]{\left\{#1\right\}}

\newcommand{\sep}{\ast}

\makeatletter
\providecommand*{\cupdot}{%
  \mathbin{%
    \mathpalette\@cupdot{}%
  }%
}
\newcommand*{\@cupdot}[2]{%
  \ooalign{%
    $\m@th#1\cup$\cr
    \hidewidth$\m@th#1\cdot$\hidewidth
  }%
}
\makeatother

\newtcbox{\tracebox}[1][]{on line,size=fbox,#1}

\iftutex
  \usepackage{unicode-math}
  \setmathfontface\altgrfont{AntykwaTorunskaMed-Italic.otf}[Scale=MatchLowercase]
  \newcommand{\kronecker}{\altgrfont{δ}}
\else
  \usepackage[T1]{fontenc}
  \usepackage[utf8]{inputenc} 
  \usepackage{amsmath}
  \DeclareSymbolFont{altgr}{OML}{antt}{m}{it}
  \DeclareMathSymbol{\kronecker}{\mathord}{altgr}{"0E}
\fi

\newcommand{\ifext}[2]{\ifdefined\extflag{#1}\else{#2}\fi}

\ifext{
  \settopmatter{printfolios=true,printccs=false,printacmref=false}
}{
  \settopmatter{printfolios=true,printccs=true,printacmref=true}
}

\setlist[itemize]{leftmargin=*}
\setlist[enumerate]{leftmargin=*}

\allowdisplaybreaks

\begin{document}

\newcommand{\mytitle}{Grammar Repair with Examples and Tree Automata}

\newcommand{\runningtitle}{\mytitle}

\setlength\floatsep{1.25\baselineskip plus 3pt minus 2pt}
\setlength\textfloatsep{1.25\baselineskip plus 3pt minus 2pt}
\setlength\intextsep{1.25\baselineskip plus 3pt minus 2 pt}


\ifext{
  \title[\mytitle]{\mytitle\\[2pt]\large Extended version (includes appendix)}
}{
  \title[\mytitle]{\mytitle}
}

\author{Yunjeong Lee}
\affiliation{%
  \institution{National University of Singapore}
    \country{Singapore}
}
\email{yunjeong.lee@u.nus.edu}
\orcid{0000-0002-3803-9782}

\author{Gokul Rajiv}
\affiliation{%
  \institution{National University of Singapore}
    \country{Singapore}
}
\email{grajiv@u.nus.edu}
\orcid{0009-0003-6074-1847}

\author{Ilya Sergey}
\affiliation{%
  \institution{National University of Singapore}
    \country{Singapore}
}
\email{ilya@nus.edu.sg}
\orcid{0000-0003-4250-5392}

\begin{abstract}
%
Context-free grammars (CFGs) are the de-facto formalism for
declaratively describing concrete syntax for programming languages and
generating parsers.
%
%
One of the major challenges in defining a desired syntax is ruling out
all possible ambiguities in the CFG productions that determine scoping
rules as well as operator precedence and associativity.
%
%
Practical tools for parser generation typically apply ad-hoc
approaches for resolving such ambiguities, which might result in a
parser's behavior that contradicts the intents of the language
designer.
%
%
In this work, we present a user-friendly approach to soundly
\emph{repair} grammars with ambiguities, which is inspired by the
\emph{programming by example} line of research in automated program
synthesis.
%
%
At the heart of our approach is the interpretation of both the initial
CFG and additional examples that define the desired restrictions in
precedence and associativity, as \emph{tree automata} (TAs).
%
%
The technical novelties of our approach are (1)~a new TA learning
algorithm that constructs an automaton based on the original grammar
and examples that encode the user's preferred ways of resolving
ambiguities all in a single TA, and (2)~an efficient algorithm for 
TA intersection that utilises reachability analysis and optimizations
that significantly reduce the size of the resulting automaton, which
results in idiomatic CFGs amenable to parser generators.
%
%
We have proven the soundness of the algorithms, and implemented our
approach in a tool called \tool, demonstrating its effectiveness on a
series of case studies.

\end{abstract}




\maketitle


\vspace{-5pt}

\section{Introduction}
\label{sec:intro}

In recent years, substantial progress has been made on
automatically synthesising grammars and parsers for context-free
languages~\cite{Bastani0AL17,LeungSL15,BettscheiderZ24,GopinathMZ20}.
Nevertheless, writing a precise description of a programming
language's concrete syntax still remains a challenging and
difficult-to-automate task.
The main challenge stems from the multitude of possibilities to
introduce ambiguities in the interpretation of language strings as
syntax trees when defining the language's context-free grammar (CFG).
For instance, according to the following CFG of a language of
arithmetic expressions
\begin{equation*}
\label{eq:cfg1}
\nonts \to \nonts + \nonts~|~\nonts \ast \nonts~|~(\nonts)~|~x~|~y~|~z   
\end{equation*}
the string $x + y * z$ can be parsed both as $(x + y) * z$ and
$x + (y * z)$, even though only one of those interpretations is
typically desired by the language designers.

Modern frameworks for LR($k$) parser generation, such as
\yacc~\cite{Johnson:yacc75}, \beaver~\cite{beaver},
\tname{ScalaBison}~\cite{Boyland-Spiewak}, and
\menhir~\cite{Regis:Menhir2016} can detect an overapproximation of 
such ambiguities in operator precedence, associativity, and nesting, 
in the form of so-called shift/reduce and reduce/reduce conflicts, 
reporting them to the users and even resolving them automatically.
Even though this simplifies the development of a language's syntax,
automated ambiguity resolution might result in an interpretation of
the syntax that is different from what is implicitly envisioned by its
designer.

To provide more control over operator precedence, associativity, and
nesting, existing tools for parser generation offer mechanisms to
specify these properties explicitly~\cite{BrandSVV02,MacedoS20}, while
standard compiler textbooks provide general strategies to describe a
CFG to avoid ambiguities in the first
place~\cite{Klint:ASMICS94,Wharton76}.
That said, adopting tool-specific conventions and following ``good
practices'' for structuring CFGs often leads to grammars that are
difficult to understand and maintain.
Even worse, should an ambiguity be introduced in a CFG, expert
knowledge, both in the structure of the object language and in the
workings of the parser generator tool is required to correctly
resolve~it.

A more \emph{declarative} approach to resolve ambiguities in a
context-free grammar has been proposed by Adams and
Might~\cite{Adams-Might:OOPSLA17} who suggested capturing the
restrictions imposed on top of an ambiguous grammar in the form of
\emph{tree regular expressions} (TREs) that explicitly forbid
undesired classes of parse trees.
The approach of Adams and Might exploits the fundamental connection
between CFGs, tree-regular expressions, and tree automata (TAs), using
the latter language representation, which is closed under
intersection, as a way to produce the \emph{repaired} (\ie
ambiguity-free) version of the grammar.
Unfortunately, writing a correct tree-regular expression that resolves
an ambiguity is not an easier task for a non-expert than fixing the
ambiguity directly in the grammar, as it requires one to have a good
grasp of TRE semantics to design suitable restrictions.
%
The goal of this work is to enhance the tree automata-based approach
to grammar repair and provide a user-friendly and sound way to repair
grammars from ambiguities by adopting a popular \emph{programming by
  example} paradigm from the works on automated program
synthesis~\cite{Gulwani11,LeG14,WangCB17,WuLF15,PelegSY18,LeCLGV17}.

\vspace{-5pt}

\paragraph{Key idea}
Our novel approach, dubbed \emph{grammar repair by example}, resolves
ambiguities in a grammar by suggesting to the user pairs of examples
that demonstrate mutually-exclusive ways to resolve parsing conflicts,
asking the user to choose one of them, and using the chosen examples
to generate a ``fixed'' version of the grammar.
More specifically, our approach (a)~converts an ambiguity detected as
an LR($1$) parser generation conflict, into a small set of concrete
examples that are shown to the user in the form of parse tree
alternatives.
Out of those examples, (b)~the user chooses the examples corresponding
to the desired syntax of the language in question.
The chosen examples are then (c)~automatically converted into a
grammar restriction in the form of a tree automaton, which is next
(d)~intersected with the automaton corresponding to the original
grammar, thus eliminating the ambiguity. Finally, (e)~the result of
the intersection is converted back to the CFG form, which is then
analysed for ambiguities again, and, in case any are detected, the
steps (a)-(e) repeat.

\vspace{-3pt}

\paragraph{Challenges}
The technical problems that needed to be solved in order to implement
this approach in practice had to do with designing algorithms for the
steps~(c) and~(d).
In particular, while step (c) appears to be a textbook automata
learning task, standard Angluin-style
algorithms~\cite{Angluin:IC81,Angluin:IC87,Besombes-Marion:TCS07} are
unsuitable for our scenario, as they require construction of tree
examples for \emph{all} the alphabet symbols.
In other words, such examples would have to collectively represent the
\emph{entire grammar} rather than its \emph{subsets} that are directly
involved in ambiguities---which is what we aim to provide to the
user of our approach.
%
%
%
%
Furthermore, using the standard definition of tree automata
intersection from the existing approaches as a way to update the
grammar with the learned
``fixes''~\cite{Adams-Might:OOPSLA17,Comon-et-al:tata08} \emph{as is}
produces non-idiomatic grammars that are difficult to comprehend.
%

We addressed these challenges by implementing two novel algorithms:
for tree automata learning from example sub-trees of an input ``base''
grammar (for the step~(c)) and for TA intersection (for~(d)). Finally,
we implemented the described end-to-end grammar repair pipeline in a
tool.

\vspace{-3pt}

\paragraph*{Contributions}

In summary, we make the following contributions:

\vspace{-3pt}

\begin{itemize}
\item Our main conceptual contribution is \emph{grammar repair by
    example}---a novel approach to automatically resolve ambiguities
  in context-free grammars following a simple input from the user
  regarding what parse (sub)trees are acceptable. The workings of our
  approach rely on a fundamental relation between context-free
  grammars and tree automata (\autoref{sec:overview}).
\item Our first technical contribution is a novel algorithm for
  passive learning of tree automata from positive/negative parse
  subtree examples and its soundness proof stating that the
  synthesised automaton indeed correctly accepts all positive and
  rejects all negative examples (\autoref{subsec:learner}).
\item Our second technical contribution is a new algorithm for
  computing an intersection of tree automata tailored to regular tree
  grammars producing a result that can be rendered as an idiomatic
  CFG, along with the proof of its correctness
  (\autoref{subsec:intersect}).
\item We implemented our approach to grammar repair by example in a
  tool called \tool on top of \menhir---a framework for parser
  generators in OCaml~\cite{Regis:Menhir2016}.
  Our evaluation on examples from undergraduate compiler classes,
  questions posed on StackOverflow, and grammars of real-world
  languages
  %
  %
  demonstrates utility and efficiency of our approach: it does indeed
  fix ambiguities, with minimal help from the user, in most of the
  case studies and does so quite fast (\autoref{sec:eval}).
\end{itemize}


\section{Overview}
\label{sec:overview}


This section illustrates an example run of \tool, demonstrating how
ambiguities in a CFG can be resolved through a lightweight interaction
with the user, by learning the ``fixes'' in the form of tree automata
(TA) from user-provided examples and subsequently updating the grammar
by means of TA intersection.
We start from a characteristic example of an initial input CFG with
ambiguities (\autoref{subsec:ambig-cfg}), explain how it is translated
into a TA (\autoref{subsec:cfg-as-ta}), describe generation of a TA
based on the user-specified tree examples and the input CFG
(\autoref{subsec:restrictions-as-ta}), and show how the ambiguities
are resolved by intersecting the two TAs and translating the result
back to a CFG (\autoref{subsec:ta-intersect}).

%
\subsection{Context-Free Grammars and Ambiguities}
\label{subsec:ambig-cfg}

A context-free grammar (CFG) is formally defined as a tuple
$(V, \Sigma, S, P)$, where $V$ is a set of nonterminal symbols,
$\Sigma$ a set of terminal symbols, $S \in V$ a start nonterminal,
and $P \subseteq V \times (V \cup \Sigma)^*$ is a set of productions.
An example of a CFG \capcal{G} is shown in~\autoref{fig:cfg}; its
nonterminals $V$ consists of \token{stmt}, \token{decl}, \token{expr}
and \token{ident}, with a start nonterminal \token{stmt}; its
terminals $\Sigma$ include \token{SEMI}, \token{TINT}, \token{IDENT},
and \token{EQ} involved in declaration statements, \token{IF},
\token{THEN}, and \token{ELSE} for conditional statements with only a
then-branch or both then- and else-branches, \token{PLUS},
\token{STAR}, and \token{INT} for expressions with binary operations,
and \token{LPAREN} and \token{RPAREN} for open and close parentheses;
its production rules \texttt{P} shown in~\autoref{fig:cfg} replace
nonterminal symbols with sequences of nonterminal and/or terminal
symbols.

\begin{figure}[tbp!]
\begin{lstlisting}[belowskip=-0.1\baselineskip,basicstyle=\linespread{1.0}\ttfamily\footnotesize,numbers=none]
V = { stmt, decl, expr, ident }
Σ = { SEMI, IF, THEN, ELSE, PLUS, STAR, INT, LPAREN, RPAREN, TINT, EQ }
S = stmt
P = 
\end{lstlisting}
\begin{minipage}{0.51\textwidth}
\begin{lstlisting}[basicstyle=\linespread{1.0}\ttfamily\footnotesize,numbers=none]
{ stmt -> decl SEMI 
  stmt -> IF expr THEN stmt
  stmt -> IF expr THEN stmt ELSE stmt 
  decl -> TINT ident EQ expr 
  ident -> IDENT
\end{lstlisting}
\end{minipage}
\begin{minipage}{0.47\textwidth}
\begin{lstlisting}[basicstyle=\linespread{1.0}\ttfamily\footnotesize,numbers=none,escapechar=!]
expr -> expr PLUS expr 
expr -> expr STAR expr 
expr -> INT 
expr -> LPAREN expr RPAREN 
expr -> ident               }
\end{lstlisting}
\end{minipage}
\caption{An example CFG \capcal{G} with ambiguities.}
\vspace{-8pt}
\label{fig:cfg}
\end{figure}

\begin{figure}[!b]
\begin{lstlisting}[basicstyle=\linespread{1.0}\ttfamily\footnotesize,numbers=none]
(1) expr PLUS expr PLUS expr
(2) expr STAR expr STAR expr
(3) expr PLUS expr STAR expr
(4) IF expr THEN IF expr THEN stmt ELSE stmt
\end{lstlisting}
\caption{Expressions demonstrating the ambiguities in the grammar from
  \autoref{fig:cfg}.}
\label{fig:ambig-exprs}
\end{figure}

\begin{figure}[t]
  \centering
     \begin{subfigure}[b]{0.48\linewidth}
         \centering
         \includegraphics[width=\linewidth]{figs/ui-interact/ui-interact-a.pdf}
         \vspace{-1.2\baselineskip}
         \caption{Left- {\footnotesize\texttt{vs.}} right-associative \code{plus3}}
         \label{fig:user-interact-a}
     \end{subfigure}
     \hfill
     \begin{subfigure}[b]{0.48\linewidth}
         \centering
         \includegraphics[width=\linewidth]{figs/ui-interact/ui-interact-b.pdf}
         \vspace{-1.2\baselineskip}
         \caption{Left- {\footnotesize\texttt{vs.}} right-associative \code{star3}}
         \label{fig:user-interact-b}
     \end{subfigure}
     \hfill
     \begin{subfigure}[b]{0.48\linewidth}
         \centering
         \includegraphics[width=\linewidth]{figs/ui-interact/ui-interact-c.pdf}
         \vspace{-1.2\baselineskip}
         \caption{\code{plus3} {\footnotesize\texttt{vs.}} \code{star3}}
         \label{fig:user-interact-c}
     \end{subfigure}
     \hfill
     \begin{subfigure}[b]{0.48\linewidth}
         \centering
         \includegraphics[width=\linewidth]{figs/ui-interact/ui-interact-d.pdf}
         \vspace{-1.2\baselineskip}
         \caption{\code{if4} {\footnotesize\texttt{vs.}} \code{if6}}
         \label{fig:user-interact-d}
     \end{subfigure}
        \caption{Tree examples representing conflicts in \capcal{G}.}
        \label{fig:user-interact}
\end{figure}
%
%
%
A CFG is called \emph{ambiguous} when there are multiple ways to match
a string of terminals via its production rules.
\capcal{G} in~\autoref{fig:cfg} is an example of such grammar, with
four different ambiguities, \aka \emph{conflicts}. 
In particular, these conflicts are caused by lack of information about
associativity of \token{PLUS} and \token{STAR} operators, and about
the precedence order between \token{PLUS} and \token{STAR} as well as
between \token{IF} with \token{THEN} branch and \token{IF} with
\token{THEN} and \token{ELSE} branches.
\autoref{fig:ambig-exprs} shows expressions that are subject to the
conflicts.
The first expression \token{(1)} can be derived by either applying the
grammar's production rule \token{expr $\rightarrow$ expr PLUS expr} on
the left \token{expr PLUS expr} first and the right one next \emph{or}
the other way around, since the grammar is ambiguous on whether
\token{PLUS} is left- or right-associative.
Similarly, the second expression \token{(2)} is subject to the
ambiguity of \token{STAR} being left-associative or right-associative.
The third expression \token{(3)} is subject to the unspecified
precedence orders between \token{PLUS} and \token{STAR}.
The fourth expression \token{(4)} illustrates an ambiguity of
parsing nested \texttt{IF}-expressions with only then- or both then-
and else-branches, famously known as the \emph{dangling else}
problem~\cite{Abrahams66}.

\begin{wrapfigure}[11]{r}{0.5\textwidth}
\vspace{-12pt}
  \centering
  \includegraphics[width=0.5\textwidth]{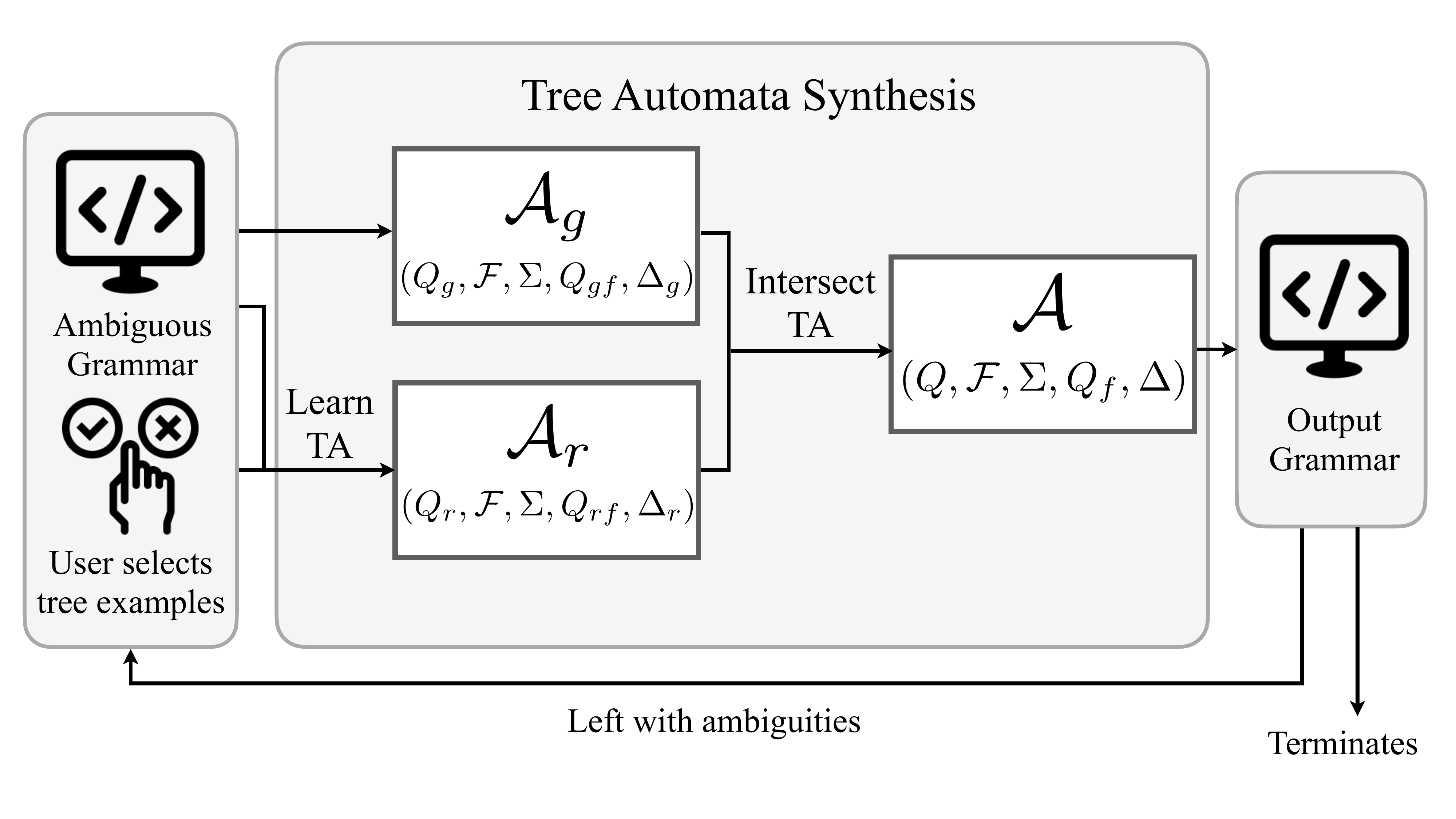}
  \setlength\abovecaptionskip{-15pt}
  \caption{\greta grammar disambiguation workflow.}
  \label{fig:workflow}
\end{wrapfigure}
These different ways to parse the same expression can be depicted via
\emph{tree examples}, as shown in~\autoref{fig:user-interact}.
%
%
%
We rely on the hierarchical property of such trees: symbols in tree
examples that are \emph{deeper} have \emph{higher} precedence orders
than symbols at lower depths.
These examples are sub-parse trees that represent ambiguities in
a CFG and thus use only a subset of the transitions.
When only one symbol is present in a tree example, we consider 
it to represent an associativity ambiguity.
%
Our idea is to allow the user to choose \emph{one such tree per
conflict}.
The resulting combined set of \emph{tree examples} that are \emph{not}
selected by the user $T^{-}$ can be then used to disambiguate the 
grammar, updating its rules accordingly.
In a nutshell, \tool achieves that by (a)~generating a TA from the
user-specified tree examples and from the initial (input) CFG and
(b)~intersecting the learned TA with a TA obtained from the original
grammar to resolve the ambiguities (\autoref{fig:workflow}).
In the rest of this section, we will walk through the stages of \tool
working using the conflict-featuring grammar \capcal{G} as an example.

%

\subsection{From a Context-Free Grammar to a Tree Automaton}
\label{subsec:cfg-as-ta}
%

%

The first step in our approach is to convert the input grammar into a
tree automaton.
It is well-known that any CFG can be represented by a tree automaton 
that recognises the language of its parse trees---so-called
\emph{regular tree
  language}~\cite{Comon-et-al:tata08,Engelfriet:arXiv15}.
We now provide a brief explanation of basic concepts of TAs and
show how the grammar from \autoref{fig:cfg} is translated to its
corresponding TA.

\begin{figure}[b]
\vspace{-8pt}
\begin{lstlisting}[basicstyle=\linespread{1.0}\ttfamily\footnotesize,numbers=none]
  (*@$\mathcal{F}$@*) = { (SEMI,2), (IF,4), (IF,6), (TINT,4), (PLUS,3), (STAR,3), (INT,1)
  (*@\hspace{3em}\,\,\,\,\,@*)((),3), (δ,1), (IDENT,1) (*@\,@*)}
\end{lstlisting}
\caption{Ranked alphabet $\mathcal{F}$ for the running example.}
\label{fig:ranked-alphabet}
\end{figure}

%
A \emph{tree automaton} (TA) is a tuple $(Q, \mathcal{F}, \Sigma,
Q_{\text{f}}, \Delta)$ where $Q$ refers to a set of states,
$\mathcal{F}$ is a set of constructor labels (\aka ranked alphabet),
$\Sigma$ is a set of terminal symbols, $Q_{\text{f}} \subseteq Q$ is a
set of final states, and $\Delta$ is a set of transition rules.
%
Each ranked symbol in $\mathcal{F}$ has an associated arity
($\textit{Rank}$) corresponding to the number of terminals and
nonterminals that appear on the right-hand side of a production in the
original CFG. It also has a symbol ($\textit{Sym}$) which is unique to
each production. 
A TA accepts a tree if there exists a run from the leaves to a final
state at the root, following the transition rules in $\Delta$
(described formally in the supplementary material).
%
%
%
%
%
%
\autoref{fig:ranked-alphabet} shows a readable set of ranked symbols
for~$\mathcal{G}$, chosen for use in this paper.
%
%
%
%
%
%



%
%

\begin{figure}[t]
\begin{lstlisting}[basicstyle=\linespread{1.0}\ttfamily\footnotesize,belowskip=-0.3\baselineskip,numbers=none]
Q_g = { stmt, decl, expr, ident }
(*@\,@*)Q_gf = { stmt }
Del_g = 
\end{lstlisting}
\begin{minipage}{0.55\textwidth}
\begin{lstlisting}[basicstyle=\linespread{1.0}\ttfamily\footnotesize,numbers=none]
{ stmt <-_semi     decl SEMI 
  stmt <-_if4     IF expr THEN stmt
  stmt <-_if6     IF expr THEN stmt ELSE stmt
  decl <-_tint     TINT ident EQ expr
  ident<-_id     IDENT
\end{lstlisting}
\end{minipage}
\begin{minipage}{0.43\textwidth}
\begin{lstlisting}[basicstyle=\linespread{1.0}\ttfamily\footnotesize,numbers=none]
expr<-_+s    expr PLUS expr
expr<-_*s    expr STAR expr
expr<-_N    INT
expr<-_()    LPAREN expr RPAREN
expr<-_ds    ident               }
\end{lstlisting}
\end{minipage}
\setlength\abovecaptionskip{5pt}
\setlength\belowcaptionskip{-10pt}
\caption{TA $\mathcal{A}_{\text{g}}$ reinterpreted from $\mathcal{G}$.}
\label{fig:ta-ambiguous}
\end{figure}

A TA \capcal{A\textsubscript{g}} translated from the grammar
\capcal{G} is shown in~\autoref{fig:ta-ambiguous}, combined with the
ranked alphabet from~\autoref{fig:ranked-alphabet}.
%
%
It is generated from the grammar as follows.
First, the productions of \capcal{G} are converted to ranked alphabet
symbols with their corresponding arities to result
in~\autoref{fig:ranked-alphabet}.
A nonterminal-to-nonterminal production like \token{expr $\rightarrow$ ident}
is mapped to a transition \code{expr <-_ds1- ident} using a \code{ds1}-label,
which is added to the ranked alphabet.
%
%
%
The set of nonterminal symbols $V$ and the start symbol $S$ of
\capcal{G} are mapped to a set of states \texttt{Q\textsubscript{g}}
and a singleton set of $S$ as the final accepting state
\texttt{Q\textsubscript{gf}}, respectively, as shown
in~\autoref{fig:ta-ambiguous}.
The set of terminal symbols of \capcal{G} becomes the set of terminal
symbols $\Sigma$ of the TA.
Lastly, production rules of \capcal{G} are annotated with each rule's
ranked alphabet symbol, resulting in transition rules of
\capcal{A\textsubscript{g}}.
%
%

\subsection{Generating a Tree Automaton from Tree Examples and an Input Grammar}
\label{subsec:restrictions-as-ta}
%

%
As we learned from~\autoref{subsec:ambig-cfg} ambiguities in a grammar
are caused by the lack of information about associativity and/or
precedence orders among terminal symbols.
%
%
We observe that \emph{semi-concrete} parse trees, as shown in
\autoref{fig:user-interact}, can be useful for illustrating the
options regarding associativity and precedence orders.
%
In our approach, we show such trees to the language designer, asking
them to select the options that correspond to their preferred
hierarchies among the alphabet symbols in question.
Based on the user's selections as well as the original CFG, we
generate a TA that rejects unwanted parse trees. The TA must also 
accept trees consistent with selections and with symbols in 
$\mathcal{G}$ unrelated to the ambiguities.
%
\subsubsection{Conflict Examples and Their Resolution}
\label{sec:examples}

Each of the ambiguities coming from the expressions
in~\autoref{fig:ambig-exprs} is presented to the user as parse trees
with different parsing orders, as shown
in~\autoref{fig:user-interact}.
%
%
Our approach relies on a standard LR(1) parser generator to identify
conflicts and generate the tree examples from the sets of productions
featuring the ambiguities.
%
%
%
For example, the expression \token{(1) expr PLUS expr PLUS expr}
in~\autoref{fig:ambig-exprs} can be parsed by applying a production
rule \token{expr $\rightarrow$ expr PLUS expr} first and then applying
the same rule to the left---or the right---\token{expr} on the
right-hand side of the production, essentially exemplifying the left-
or the right-associativity of the \code{plus3} operator.
Once the rule is labeled with its respective tree-constructor label,
\code{plus3} in this case, the trees representing different ways of
parsing are constructed, as shown in~\autoref{fig:user-interact-a}.
Likewise, trees representing left- and right-associativity of \code{star3}
in~\autoref{fig:user-interact-b} represent two different ways to parse
\token{(2) expr STAR expr STAR expr} in~\autoref{fig:ambig-exprs}. 
%
As observed in these sets of productions, if a conflict involves
just one symbol, we can infer that it is because of the symbol's
unclear associativity.

Other than associativity, the conflicts might be caused by ambiguous
precedence orders among different symbols.
Consider \token{(3) expr PLUS expr STAR expr}
in~\autoref{fig:ambig-exprs}.
The expression can be parsed by applying the production \token{expr
  $\rightarrow$ expr PLUS expr} first and then \token{expr
  $\rightarrow$ expr STAR expr} to the right \token{expr}, generating
a parse tree on the left in~\autoref{fig:user-interact-c}.
Alternatively, \token{expr $\rightarrow$ expr STAR expr} can be
applied first and then \token{expr $\rightarrow$ expr PLUS expr} to
the left \token{expr}, generating a tree on the right
in~\autoref{fig:user-interact-c}.
These represent two alternative ways to generate parse trees 
by putting these symbols at different \emph{depths} in those
trees when precedence orders between these symbols (\eg,
\code{plus3} and \code{star3}) are ambiguous.
Similarly, ambiguity from parsing the expression \token{(4) IF expr
  THEN IF expr THEN stmt ELSE stmt} in~\autoref{fig:ambig-exprs},
known as the \emph{dangling else problem}, is shown as trees
in~\autoref{fig:user-interact-d}, where the left tree symbolises
\code{if4} having a higher precedence order than \code{if6} and the
right one with an opposite precedence relation.

Coming back to our running example, suppose the user has indicated
their parsing preferences by selecting the ticked trees
in~\autoref{fig:user-interact}.
This results in \tool learning the relations of the tree-constructor
labels \code{plus3} $<$ \code{star3} and \code{if4} $<$ \code{if6} in
addition to \code{plus3} and \code{star3} being left-associative.
These user preferences can be thought of as \emph{restrictions} on the
input grammar.
%
%
\tool collects these restrictions by traversing each of the
trees \emph{not} selected by the user to combine them with the 
so-called \emph{base precedence order} of the original grammar.

The base precedence order $O_{\text{bp}}$ denotes the set of all the
ranked alphabet symbols obtained from the tree
automaton~$\mathcal{A}_{\text{g}}$ representing the input CFG~$\mathcal{G}$ with
their corresponding orders, which can be thought of as the
\emph{shortest distance} to the final accepting state of
$\mathcal{A}_{\text{g}}$ in~\autoref{fig:ta-ambiguous}.
For our example whose ranked alphabet denoted by $\mathcal{F}$ is
shown in~\autoref{fig:ranked-alphabet}, $O_{\text{bp}}$ is as follows:
\begin{lstlisting}[basicstyle=\linespread{1.0}\ttfamily\footnotesize,numbers=none]
{(*@\,@*)((IF,4),0), ((IF,6),0), ((SEMI,2),0), ((TINT,4),1), ((PLUS,3),1), ((STAR,3),1), 
 (*@\,@*)((INT,1),1), (((),3),1), ((δ,1),1)(*@\,@*)}
\end{lstlisting}

The $1$-arity symbol \code{IDENT} is identified as a \emph{trivial
symbol} (detailed in section \ref{subsec:base-precedence-order}) and 
is handled separately, since we know that it will not be conflicting 
with other symbols, and their corresponding nonterminals remain 
intact throughout the run of \tool. Its associated state  \code{ident} 
as well as transition \code{ident <-_ident1- IDENT} remain unchanged 
in the generated TA.
%
%
%
%
%
%
We provide a formal definition of $O_{\text{bp}}$ and describe a
procedure to construct it in~\autoref{subsec:learner}.

%
%
\subsubsection{Parsing Preferences as a Tree Automaton}
\label{sec:prefs}

The following step---constructing a tree automaton that elaborates the
input grammar with additional restrictions---is the key novel idea of
this work.

\begin{figure}[t!]
\begin{lstlisting}[basicstyle=\linespread{1.0}\ttfamily\footnotesize,belowskip=-0.3\baselineskip,numbers=none]
Q_r = { e_0, e_1, e_2, e_3, e_4, ident }
(*@\,@*)Q_rf = { e_0 }
Del_r = 
\end{lstlisting}
\centering
\begin{subfigure}{0.49\linewidth}
\begin{lstlisting}[basicstyle=\linespread{1.0}\ttfamily\footnotesize,numbers=none]
{ e_0 <-_if4    IF e_0 THEN e_0
  e_0 <-_semi    e_0 SEMI
  e_0 <-_εs    e_1
  e_1 <-_if6    IF e_1 THEN e_1 ELSE e_1
  e_1 <-_semi    e_1 SEMI
  e_1 <-_εs    e_2
  e_2 <-_+s    e_2 PLUS e_3
  e_2 <-_tint    TINT e_2 EQ e_2
  e_2 <-_N    INT
  e_2 <-_()    LPAREN e_2 RPAREN
  e_2 <-_ds    ident
  e_2 <-_εs    e_3
\end{lstlisting}
\end{subfigure}
\begin{subfigure}{0.49\linewidth}
\begin{lstlisting}[basicstyle=\linespread{1.0}\ttfamily\footnotesize,numbers=none]
e_3 <-_*s    e_3 STAR e_4
e_3 <-_tint    TINT ident EQ e_3
e_3 <-_N    INT
e_3 <-_()    LPAREN e_3 RPAREN
e_3 <-_ds    ident
e_3 <-_εs    e_4
e_4 <-_tint    TINT ident EQ e_4
e_4 <-_N    INT
e_4 <-_()    LPAREN e_4 RPAREN
e_4 <-_ds    ident
ident <-_id    IDENT
(*@@*)                             (*@\,@*)}
\end{lstlisting}
\end{subfigure}
\caption{The tree automaton $\mathcal{A}_{\text{\text{r}}}$ generated from
  user-specified parsing preferences.}
\label{fig:ta-restrictions}
\end{figure}

Notice that the trees in \autoref{fig:user-interact} are \emph{not}
parse trees of the input grammar~$\mathcal{G}$; rather they
are \emph{sub-trees} of some of the parsing trees allowed by the
grammar. To wit, they feature symbols \code{plus3}, \code{star3},
\code{if4}, and \code{if6}, while missing information about precedence
orders of other symbols from the alphabet $\mathcal{F}$, such as
\code{semi2} or \code{tint4} with respect to symbols in the examples,
as that information can be restored via $O_{\text{bp}}$. Given the
preferences indicated by the user on the provided examples, our goal
is to capture them into a \emph{new} tree
automaton~\capcal{A\textsubscript{\text{r}}} that is also ``permissive
enough'' to accept the desired ``complete'' parse trees allowed by the
TA $\mathcal{A}\textsubscript{g}$ corresponding to the original
grammar. In the remainder of this section, we provide an informal
description of this procedure, fleshing out its details and correctness 
argument in \autoref{subsec:learner}.

We start by combining the restrictions learned from the examples
in~\autoref{fig:user-interact} (\eg, \code{plus3} $<$ \code{star3},
\code{if4} $<$ \code{if6}, \etc) with the base precedence order
$O_{\text{bp}}$, obtaining the following set of precedence
orders---let's call it $O_{p}$---for all symbols:
\begin{lstlisting}[basicstyle=\linespread{1.0}\ttfamily\footnotesize,numbers=none]
{(*@\,@*)((IF,4),0), ((SEMI,2),0), ((IF,6),1), ((SEMI,2),1), ((PLUS,3),2), ((TINT,4),2),
 ((INT,1),2), (((),3),2), ((δ,1),2), ((STAR,3),3), ((TINT,4),3), ((INT,1),3), 
 (((),3),3), ((δ,1),3), ((TINT,4),4), ((INT,1),4), (((),3),4), ((δ,1),4)(*@\,@*)}
\end{lstlisting}
%
%
When orders of symbols from the chosen examples are compared in this
updated set $O_{p}$, it is consistent with the precedence relations
\code{plus3} $<$ \code{star3} as well as \code{if4} $<$ \code{if6}.
In addition, for all other pairs of symbols in the base precedence 
order, there exists a pair in $O_p$ for those symbols with the 
same relative order.
$O_p$ also accommodates for later integration with associativity 
restrictions $O_a$, full details of which are explained in 
\autoref{subsec:learner}.

%
Subsequently, given a pair of ranked symbol and its order $(f, o)$ in
$O_{p}$, \tool generates an $f$-labeled transition from the 
states and terminals in the right-hand side of $\textit{Prod}(f)$ to 
\texttt{e}\textsubscript{$o$}. 
As an example, \code{((IF,4),0)} from $O_{p}$ produces the
transition \code{e_0 <-_if4- *if* e_0 *then* e_0}.
At the same time, each ordered state \texttt{e}\textsubscript{$o$} is
linked to its next ordered state in hierarchy
\texttt{e}\textsubscript{$o+1$} by an $\epsilon$-transition: \eg,
\code{e_0 <-_eps1- e_1}.
The state \code{e_0} corresponding to the symbol \code{(IF,4)}
represents a level in the hierarchy of parsing orders. 
%
Moreover, if a symbol was used to specify an associativity in the
examples, \tool generates a transition that takes it into account:
\eg, \code{e_2 <-_plus3- e_2 *plus* e_3} produced based on
left-associative \code{(PLUS,3)} from~$O_{a}$ and \code{((PLUS,3),2)}
from~$O_{p}$.
%
%
These steps lead to construction of
$\mathcal{A}\textsubscript{\text{r}}$
in~\autoref{fig:ta-restrictions}.
The details of the algorithm are explained in \autoref{subsec:learner}.


%
Notice that the automaton obtained this way is not immediately the
result we want that corresponds to the repaired grammar.
That is because the constructed $\mathcal{A}\textsubscript{\text{r}}$
recognises not only all those non examples-related parse trees allowed
by $\mathcal{A}\textsubscript{\text{g}}$ but \emph{even} those
terms that are not allowed by $\mathcal{A}\textsubscript{\text{g}}$. 
For example, consider an expression \code{*if* 1 THEN 2 STAR 3}. A
tree corresponding to this expression is accepted by
$\mathcal{A}\textsubscript{\text{r}}$ since \code{e_1} appearing after
\code{THEN} can be rewritten by \code{e_3 STAR e_4}, whereas this is
not possible in the original grammar, hence \emph{not} accepted by
$\mathcal{A}\textsubscript{\text{g}}$.
This is the reason why we need to take an intersection of the tree
automata to eventually produce a grammar that is both \emph{not} too
permissive \wrt original grammar and restrictive \wrt user-specified
parsing preferences.



\begin{figure}[t!]
\begin{minipage}{0.55\textwidth}
\begin{lstlisting}[style=HL,numbers=none]
!st_e0(*@\,@*)<-_semi2-(*@\,@*)de_e0 SEMI(*@\\@*)
st_e0(*@\,@*)<-_semi2-(*@\,@*)de_e1 SEMI(*@\\@*)
st_e0(*@\,@*)<-_If4(*@\,@*)-if(*@\,@*)ex_e0 -then(*@@*)st_e0(*@\\@*)
st_e0(*@\,@*)<-_If6(*@\,@*)-if(*@\,@*)ex_e1 -then(*@@*)st_e1 -else st_e1!

de_e0(*@\,@*)<-_tint4-(*@@*)-tint(*@@*)id_id(*@\,\,@*)-eq(*@@*)ex_e2
de_e0(*@\,@*)<-_tint4-(*@@*)-tint(*@@*)id_id(*@\,\,@*)-eq(*@@*)ex_e3
de_e0(*@\,@*)<-_tint4-(*@@*)-tint(*@@*)id_id(*@\,\,@*)-eq(*@@*)ex_e4

de_e1(*@\,@*)<-_tint4-(*@@*)-tint(*@@*)id_id(*@\,\,@*)-eq(*@@*)ex_e2
de_e1(*@\,@*)<-_tint4-(*@@*)-tint(*@@*)id_id(*@\,\,@*)-eq(*@@*)ex_e3
de_e1(*@\,@*)<-_tint4-(*@@*)-tint(*@@*)id_id(*@\,\,@*)-eq(*@@*)ex_e4

`ex_e0(*@\,@*)<-_+(*@\,@*)ex_e2(*@\,@*)-plus(*@\,@*)ex_e3`
`ex_e0(*@\,@*)<-_int(*@\,@*)INT`
`ex_e0(*@\,@*)<-_d(*@\,@*)id_id`
`ex_e0(*@\,@*)<-_(-)(*@\,@*)-lparen (*@@*)ex_e2 -rparen`
`ex_e0(*@\,@*)<-_(-)(*@\,@*)-lparen (*@@*)ex_e3 -rparen`
`ex_e0(*@\,@*)<-_(-)(*@\,@*)-lparen (*@@*)ex_e4 -rparen`

`ex_e1(*@\,@*)<-_+(*@\,@*)ex_e2(*@\,@*)-plus(*@\,@*)ex_e3`
`ex_e1(*@\,@*)<-_int(*@\,@*)INT`
`ex_e1(*@\,@*)<-_d(*@\,@*)id_id`
`ex_e1(*@\,@*)<-_(-)(*@\,@*)-lparen (*@@*)ex_e2 -rparen`
`ex_e1(*@\,@*)<-_(-)(*@\,@*)-lparen (*@@*)ex_e3 -rparen`
`ex_e1(*@\,@*)<-_(-)(*@\,@*)-lparen (*@@*)ex_e4 -rparen`

st_e1(*@\,@*)<-_semi2-(*@\,@*)de_e1 SEMI
st_e1(*@\,@*)<-_If6(*@\,@*)-if(*@\,@*)ex_e1(*@\,@*)-then(*@@*)st_e1 -else st_e1

id_id <-_id-(*@\,@*)INT

`ex_e2(*@\,@*)<-_+(*@\,@*)ex_e2(*@\,@*)-plus(*@\,@*)ex_e3`
`ex_e2(*@\,@*)<-_int(*@\,@*)INT`
`ex_e2(*@\,@*)<-_d(*@\,@*)id_id`
`ex_e2(*@\,@*)<-_(-)(*@\,@*)-lparen (*@@*)ex_e2 -rparen`
`ex_e2(*@\,@*)<-_(-)(*@\,@*)-lparen (*@@*)ex_e3 -rparen`
`ex_e2(*@\,@*)<-_(-)(*@\,@*)-lparen (*@@*)ex_e4 -rparen`

ex_e3(*@\,@*)<-_*(*@\,@*)ex_e3(*@\,@*)-star(*@\,@*)ex_e4
ex_e3(*@\,@*)<-_int(*@\,@*)INT
ex_e3(*@\,@*)<-_d(*@\,@*)id_id
ex_e3(*@\,@*)<-_(-)(*@\,@*)-lparen (*@@*)ex_e3 -rparen
ex_e3(*@\,@*)<-_(-)(*@\,@*)-lparen (*@@*)ex_e4 -rparen

ex_e4(*@\,@*)<-_int(*@\,@*)INT
ex_e4(*@\,@*)<-_d(*@\,@*)id_id
ex_e4(*@\,@*)<-_(-)(*@\,@*)-lparen (*@@*)ex_e4 -rparen
\end{lstlisting}
\subcaption{\footnotesize{Cross product result of $\Delta_{\text{g}}$ and $\Delta_{\text{r}}$.}}
\label{fig:intersect-a}
\end{minipage}
\begin{minipage}{0.42\textwidth}
\begin{lstlisting}[style=HL,numbers=none]
(*@@*)
(*@@*)
(*@@*)
(*@@*)
(*@@*)
(*@@*)
(*@@*)
(*@@*)
(*@@*)
st0(*@\,@*)<-_semi2-(*@\,@*)decl SEMI
st0(*@\,@*)<-_semi2-(*@\,@*)decl SEMI
st0(*@\,@*)<-_If4(*@\,@*)-if(*@\,@*)ex0 -then(*@@*)st0
st0(*@\,@*)<-_If6(*@\,@*)-if(*@\,@*)ex1 -then(*@@*)st1 -else st1

decl(*@\,@*)<-_tint4-(*@@*)-tint(*@@*)ident(*@\,\,@*)-eq(*@@*)ex0
decl(*@\,@*)<-_tint4-(*@@*)-tint(*@@*)ident(*@\,\,@*)-eq(*@@*)ex1
decl(*@\,@*)<-_tint4-(*@@*)-tint(*@@*)ident(*@\,\,@*)-eq(*@@*)ex2

st1(*@\,@*)<-_semi2-(*@\,@*)decl SEMI
st1(*@\,@*)<-_If6(*@\,@*)-if(*@\,@*)ex1(*@\,@*)-then(*@@*)st1 -else st1

ex0(*@\,@*)<-_+(*@\,@*)ex2(*@\,@*)-plus(*@\,@*)ex1
ex0(*@\,@*)<-_(-)(*@\,@*)-lparen (*@@*)ex0 -rparen
ex0(*@\,@*)<-_(-)(*@\,@*)-lparen (*@@*)ex1 -rparen
(*@\tikzmark{a}@*)ex0(*@\,@*)<-_int(*@\,@*)INT
ex0(*@\,@*)<-_d(*@\,@*)ident
ex0(*@\,@*)<-_(-)(*@\,@*)-lparen (*@@*)ex2 -rparen (*@\tikzmark{b}@*)

ex1(*@\,@*)<-_*(*@\,@*)ex1(*@\,@*)-star(*@\,@*)ex2
ex1(*@\,@*)<-_(-)(*@\,@*)-lparen (*@@*)ex1 -rparen
(*@\tikzmark{c}@*)ex1(*@\,@*)<-_int(*@\,@*)INT
ex1(*@\,@*)<-_d(*@\,@*)ident
ex1(*@\,@*)<-_(-)(*@\,@*)-lparen (*@@*)ex2 -rparen (*@\tikzmark{d}@*)

(*@\tikzmark{e}@*)ex2(*@\,@*)<-_int(*@\,@*)INT
ex2(*@\,@*)<-_d(*@\,@*)ident
ex2(*@\,@*)<-_(-)(*@\,@*)-lparen (*@@*)ex2 -rparen (*@\tikzmark{f}@*)

ident <-_id-(*@\,@*)IDENT
(*@@*)
(*@@*)
(*@@*)
(*@@*)
(*@@*)
(*@@*)
(*@@*)
(*@@*)
(*@@*)
\end{lstlisting}
\subcaption{\footnotesize{After removing duplicates and renaming states.}}
\label{fig:intersect-b}
\vspace{0.5\baselineskip}
\end{minipage}
\setlength\abovecaptionskip{5pt}
\setlength\belowcaptionskip{-20pt}
\caption{Minimisation of the cross product of transitions $\Delta_{\text{g}}$ and $\Delta_{\text{r}}$.}

\begin{tikzpicture}[remember picture,overlay]
\draw[hlBoxDotted]
  ([shift={(-3pt,1.3ex)}]pic cs:a) 
    rectangle 
  ([shift={(3pt,-0.65ex)}]pic cs:b);
\end{tikzpicture} 
\begin{tikzpicture}[remember picture,overlay]
\draw[hlBoxDotted]
  ([shift={(-3pt,1.5ex)}]pic cs:c) 
    rectangle 
  ([shift={(3pt,-0.65ex)}]pic cs:d);
\end{tikzpicture} 
\begin{tikzpicture}[remember picture,overlay]
\draw[hlBoxDotted]
  ([shift={(-3pt,1.5ex)}]pic cs:e) 
    rectangle 
  ([shift={(3pt,-0.65ex)}]pic cs:f);
\end{tikzpicture} 

\label{fig:intersect-deltas}
\end{figure}

\subsection{Repairing the Grammar by Intersecting Tree Automata}
\label{subsec:ta-intersect}
%

We now have two tree automata:
$\mathcal{A\textsubscript{g}} \eqdef (Q_{\text{g}}, \mathcal{F}, \Sigma, Q_{\text{gf}},
\Delta_{\text{g}})$ from \capcal{G} and
$\mathcal{A\textsubscript{\text{r}}} \eqdef (Q_{\text{r}}, \mathcal{F}, \Sigma, Q_{\text{rf}}, \Delta_{\text{r}})$ learned from the tree examples as well as \capcal{G}.
The outcome of our example-based ambiguity repair is captured by the
tree automaton $\mathcal{A}_{\text{res}} \eqdef (Q, \mathcal{F}, \Sigma, Q_{\text{f}}, \Delta)$ obtained as an
intersection of $\mathcal{A\textsubscript{g}}$ and
$\mathcal{A\textsubscript{\text{r}}}$.
It is defined as a tuple that consists of cross-products of each
component~\cite{Comon-et-al:tata08}: \ie, $Q = Q_{\text{g}} \times Q_{\text{r}}$,
$Q_{\text{f}} = Q_{\text{gf}} \times Q_{\text{rf}}$, and
$\Delta = \Delta_{\text{g}} \times \Delta_{\text{r}}$. Let us discuss each component for our example.

First, taking the product of the accepting states of
\capcal{A\textsubscript{g}} and \capcal{A\textsubscript{\text{r}}}
produces the new state \code{(stmt,e_0)}, which is the start for the
intersected TA.
%
Next, starting from the state \code{(stmt,e_0)} as the left-hand side
of the transition, going across \code{stmt}-producing transitions in
$\Delta_{\text{g}}$ and \code{e_0}-producing transitions in $\Delta_{\text{\text{r}}}$ reveals
that there are three constructor labels ---\code{semi2}, \code{if4},
and \code{if6}---with all the \emph{matching} right-hand sides.
Note \code{x}-producing transitions essentially refer to those rules
that transition to the state \code{x} in the TA.
%
Given a state $(\alpha_{\text{g}}, \alpha_{\text{\text{r}}})$, the
transitions $\alpha_{\text{g}}' \leftarrow \beta_{\text{g}} \in \Delta_{\text{g}}$ and
$\alpha_{\text{r}}' \leftarrow \beta_{\text{r}} \in \Delta_{\text{r}}$ (where 
$\alpha_{\text{g}}'$ and $\alpha_{\text{r}}'$ are $\alpha_{\text{g}}$ and $\alpha_{\text{r}}$
respectively, or reachable from them by epsilon transitions) are
considered \emph{matching} when the lengths (number of terminals and
states) in $\beta_{\text{g}}$ and in $\beta_{\text{r}}$ are identical
\emph{and} at each position, either both elements are states or are
\emph{the same} terminal.

For example, consider the transition for symbol \code{if4}. In
$\Delta_{\text{g}}$, the transition \code{stmt <-_if4-} \code{IF expr
THEN stmt}, matches transition \code{e_1 <-_if4- IF e_1 THEN e_1}
in $\Delta_{\text{r}}$ because their right-hand sides have the same
number of elements \emph{and} at each position, it is either both
states or the same terminal symbol.
%
On the other hand, when we look at the symbol \code{plus3}, we cannot
find a corresponding transition in $\Delta_{\text{g}}$ producing
\code{stmt}, taking \code{plus3} as a constructor label.
This results in no \code{(stmt,e_1)}-producing transition for the
symbol \code{plus3} in the tree automata intersection.

Following this intuition, taking cross-product of transitions producing \code{(stmt, e_0)} results in {\btHL[hlReach]transitions in
  light grey}~~with labels \code{semi2}, \code{if4}, \code{if6}
in~\autoref{fig:intersect-a}.
Moreover, it shows that there are following states that can be
\emph{reached} from the \code{(stmt,e_0)}-producing transitions:
\code{(decl,e_0)}, \code{(decl,e_1)}, \code{(expr,e_0)}, \code{(expr,e_1)},
\code{(stmt,e_0)} and \code{(stmt,e_1)}.
%
We look at these so-called \emph{reachable} states, 
%
%
identifying transitions that produce each of them, so that
considering those transitions might add new states to the reachable
states.
%
%
This generates the rest of the transition rules 
in~\autoref{fig:intersect-a}.

Once all the transitions producing the reachable states are
obtained, these rules are examined in order to get rid of any
duplicate states and accordingly the rules that transition to the
duplicate states.
In our running example, the states \code{(expr, e_0)}, \code{(expr,e_1)} 
and \code{(expr,e_2)} as well as \code{(decl, e_0)} and \code{(decl, e_1)} 
are identified as duplicates. We highlight the former {\btHL[hlDup]sets 
of transitions in darker grey}, which all have identical right-hand sides.
Hence, we remove two of them---let's say, \code{(expr,e_1)} and
\code{(expr,e_2)}---and their transition rules, while also replacing
all their occurrences with \code{(expr,e_0)}. We do this also for
\code{(decl,e_0)} and \code{(decl,e_1)}.
%
%
After removing these duplicates and renaming the states, we obtain a
smaller, yet equivalent, set of transitions
in~\autoref{fig:intersect-b}.
%

Lastly, we introduce \code{eps1}-transitions to further simplify the
resulting transitions.
%
For example, looking at the \code{ex_2}-producing transitions,
\code{ex_0}- and \code{ex_1}-producing transitions repeat them, as
indicated by the \tikzmark{g}transitions in dotted boxes\tikzmark{h}
in~\autoref{fig:intersect-b}. 
\begin{tikzpicture}[remember picture,overlay]
\draw[hlBoxDotted]
  ([shift={(-1pt,2ex)}]pic cs:g) 
    rectangle 
  ([shift={(1pt,-0.65ex)}]pic cs:h);
\end{tikzpicture} 
%
%
%
That is, the transition rules take the same set of tree-constructor
labels and respectively transition to the same set of right-hand
sides. 
Hence, we can simplify \code{ex_0}- and \code{ex_1}-producing transitions by
respectively adding the rules \code{ex_0 <-_eps1- ex_2} and \code{ex_1 <-_eps1-
ex_2}, while removing the repeated transitions.
We repeat this process for all the states and their transitions,
resulting in a final TA $\mathcal{A}\textsubscript{res}$
in~\autoref{fig:ta-intersection-result}.
The obtained tree automaton $\mathcal{A}_{\text{res}}$
(\autoref{fig:ta-intersection-result}) is converted back to its
corresponding CFG trivially, by un-labeling the
transitions.\footnote{Since $\epsilon$-transitions don't consume tree
  nodes (represented with symbols), they \textit{cannot} be trivially unlabelled, as they lead to 
  productions in trees that weren't in the original grammar. Instead, 
  the unit productions after unlabelling need to be removed, to produce
  a new CFG. In \menhir, we circumvent this with semantic actions in 
  curly braces, (\eg \texttt{expr0 -> expr1 \{ \$1 \}}) with which we can 
  avoid having such productions appear in the final AST.}
%
%
This grammar is then passed back to \tool. If there are any remaining
ambiguities in it, then \tool performs all the previous steps again.
It stops running when there are no more ambiguities in the grammar,
\textit{or} any remaining ones are outside the scope of \tool or
\menhir, details of which we discuss in~\autoref{subsec:termination}.
%
%
\tool eventually terminates because the original CFG and tree automata
constructed from it as well as its ambiguities are finite, while the
number of ambiguities (and hence the number of parse trees admitted by
the CFG) decreases with each run of \tool
(\autoref{subsec:termination}).
%
%
We present our intersection algorithm in~\autoref{subsec:intersect}.

\begin{figure}[t!]
\begin{lstlisting}[basicstyle=\linespread{1.0}\ttfamily\footnotesize,belowskip=-0.3\baselineskip,numbers=none]
Q = { stmt0, stmt1, expr0, expr1, expr2, ident }
Q_f = { stmt0 }
Del = 
\end{lstlisting}
\centering
\begin{subfigure}{0.57\linewidth}
\begin{lstlisting}[basicstyle=\linespread{1.0}\ttfamily\footnotesize,numbers=none]
{ stmt0(*@\,@*)<-_if4  (*@\,@*)IF expr0 THEN stmt0
  stmt0(*@\,@*)<-_εs (*@\,\,@*)stmt1
  stmt1(*@\,@*)<-_semi   (*@\,\,@*)decl SEMI 
  stmt1(*@\,@*)<-_if6  (*@\,@*)IF expr0 THEN stmt1 ELSE stmt1
  decl(*@\,@*)<-_tint   (*@\,@*)TINT ident EQ expr0
  ident(*@\,@*)<-_id   (*@\,\,@*)IDENT
  (*@@*)
\end{lstlisting}
\end{subfigure}
\begin{subfigure}{0.42\linewidth}
\begin{lstlisting}[basicstyle=\linespread{1.0}\ttfamily\footnotesize,numbers=none]
expr0(*@\,@*)<-_+s   (*@\,\,@*)expr0 PLUS expr1
expr0(*@\,@*)<-_εs (*@\,\,@*)expr1
expr1(*@\,@*)<-_*s   (*@\,\,@*)expr1 STAR expr2
expr1(*@\,@*)<-_εs (*@\,\,@*)expr2
expr2(*@\,@*)<-_N  (*@\,@*)INT
expr2(*@\,@*)<-_εs (*@\,\,@*)ident
expr2(*@\,@*)<-_()  (*@\,@*)LPAREN expr0 RPAREN(*@\,\,@*)}
\end{lstlisting}
\end{subfigure}
\caption{An automaton $\mathcal{A}_{\text{res}}$ resulted from intersection of
$\mathcal{A}_{\text{g}}$ and $\mathcal{A}_{\text{\text{r}}}$.}
\label{fig:ta-intersection-result}
\end{figure}


\subsection{Putting It All Together}
\label{subsec:altogether}
The overall workflow of grammar disambiguation in \greta framework works as follows, as illustrated in~\autoref{fig:workflow}. 
%
%
If an ambiguous CFG is provided as an input, it is interpreted as a
tree automaton, and the ambiguities in it are presented to the user,
asking the user, for each of them, to select one of the two
alternative tree examples, which together represent different ways to
parse the same expression per conflict (\autoref{sec:examples}).
Once the user specifies their preferences by selecting a set of the
tree examples, \tool subsequently uses the chosen examples as well as
the precedence orders of all the alphabet symbols from the initial
grammar to learn a TA encoding the user's parsing preferences in line
with the grammar (\autoref{sec:prefs}).
Next, the two TAs are intersected to result in a new TA
(\autoref{subsec:ta-intersect}), which is then translated back to its
corresponding grammar.
If the resulting grammar is still left with any ambiguities, \tool is
run on it again until the grammar is fully disambiguated or only left
with non-addressable ones.
In this process, the user is simply involved in clarifying parsing
preferences in the form of tree examples, without having to compute TA
operations or even knowing that these operations are done at all.


\section{Grammar Repair by Example, Formally}
\label{sec:methodology}


This section describes technical details of the \greta framework. 
First, we present an algorithm for learning a TA from tree examples
and a base grammar, with its soundness guarantees
in~\autoref{subsec:learner}.
Next, we explain the algorithm we use for intersecting TAs and
show its correctness~\autoref{subsec:intersect}.

\subsection{From Tree Examples to a Tree Automaton}
\label{subsec:learner}

%
%
%
%
%
The learning process can be largely divided into two parts: (a)
learning restrictions about the associativity and precedence order
through \textsc{LearnOaOp} (\autoref{algo:oaop-learner}) from the tree
examples and (b) subsequently constructing a TA via \textsc{GenTA}
procedure (\autoref{algo:ta-learner}).
Before we describe each of the algorithms, we provide a formal 
description of tree examples and tree automata.

%
Unlike a complete parse tree, tree examples can start from any
nonterminal, and can have nonterminals at the leaves instead of terminals, 
\eg $\text{\code{PLUS}}_3(\text{\code{STAR}}_3(\text{\code{expr}}, *, \text{\code{expr}}), +, \text{\code{expr}})$.
An example tree represents a set of valid complete parse trees for which
(1) all child nonterminals in the tree example are substituted with valid
subtrees rooted at that nonterminal, and (2) the resulting tree after
substitution is a subtree of the complete CFG parse tree. Given a tree example $t$, we write
$\textit{ParseTrees}(t)$ to denote the set of parse trees it
represents. Suppose we have tree languages 
$\mathcal{L}_{\text{g}}$ representing the set of parse trees of the
grammar $\mathcal{G}$ rooted at the start nonterminal,
$\mathcal{T}^{nt}_{\text{g}}$ representing parse trees rooted at the
nonterminal $nt$ in $\mathcal{G}$, and functions
$\mathcal{T}^{nt}_{\text{g}} \to \mathcal{L}_{\text{g}}$ representing
parse tree contexts for some nonterminal $nt$ in $\mathcal{G}$. Then,
tree examples with $k$ incomplete nonterminals $nt_1, \dots, nt_k$,
rooted at $nt$ are functions
$\mathcal{T}^{nt_1}_{\text{g}} \times \dots \times
\mathcal{T}^{nt_k}_{\text{g}} \to \mathcal{T}^{nt}_{\text{g}}$, and
$\textit{ParseTrees}(t)$ is defined as:
\begin{align*}
\textit{ParseTrees}(t) \eqdef \{ t' \mid {}& \exists~nt, nt_1, \dots, nt_k,\, \exists~C \in \mathcal{T}^{nt}_{\text{g}} \to \mathcal{L}_{\text{g}}, t \in \mathcal{T}^{nt_1}_{\text{g}} \times \dots \times \mathcal{T}^{nt_k}_{\text{g}} \to \mathcal{T}^{nt}_{\text{g}},\\
& \exists~t_1 \in \mathcal{T}^{nt_1}_{\text{g}}, \dots, t_k \in \mathcal{T}^{nt_k}_{\text{g}} .\; t' = C[t(t_1, \dots, t_k)] \}
\end{align*}

\tool currently supports tree examples that use exactly two
productions, which covers the vast majority of ambiguities in
practice. We define some convenient notation to work with these
examples: given a tree example $t$, $t_T$ denotes the root (top)
symbol of $t$, $t_B$ denotes the nested (bottom) symbol, $t_{idx}$
denotes the index $i$ (starting from 0 on the left) of the child node
at which $t_{B}$ appears, where $0 \leq i < \textit{Rank}(t_T)$. Given
symbols $\alpha$, $\beta$ and index $i$, $\textit{Eg}(\alpha, \beta,
i)$ denotes the tree example such that $\textit{Eg}(\alpha, \beta,
i)_T = \alpha$, $\textit{Eg}(\alpha, \beta, i)_B = \beta$,
$\textit{Eg}(\alpha, \beta, i)_{idx} = i$ and all other child nodes
are wildcards. For example,
$\textit{Eg}(\text{\code{PLUS}}_3,\text{\code{STAR}}_3, 0)$ denotes
the tree example 
$\text{\code{PLUS}}_3(\text{\code{STAR}}_3(\text{\code{expr}}, *, \text{\code{expr}}), +, \text{\code{expr}})$.
An alternative to a tree example will involve the same
two symbols in a different configuration. The tree examples supported
by \tool are the ones that can be expressed with $\textit{Eg}$.

An example $t$ is an associativity-related example if 
$t_T = t_B$, and a precedence order-related example if $t_T \neq t_B$.
When a user is presented with two tree examples, the one that is 
not selected, corresponds to trees that we want to exclude from the
repaired grammar. $P^-(t)$ denotes the set of parse trees to 
remove, corresponding to a tree example $t$ not selected by the user:
\[
P^-(t) \eqdef
\begin{cases}
\textit{ParseTrees}(t) & \text{if } t_T = t_B \\
\bigcup\limits_{0 \leq i < \textit{Rank}(t_T)} \textit{ParseTrees}(\textit{Eg}(t_T, t_B, i)) & \text{otherwise}
\end{cases}
\]

In resolving precedence order related conflicts, we enforce a 
hierarcy of symbols in which a symbol of higher precedence appears
strictly deeper in the tree than a symbol of lower precedence, hence 
requiring a union of excluded parse trees for all possible child 
positions. 
%
From the set of symbols $\mathcal{S}_{\text{C}}$ involved in precedence
or associativity related conflicts, \greta constructs the following
partition of $\mathcal{S}_{\text{C}}$:
\begin{align*}
  \mathcal{S}_{\text{E}} = \{ S \subseteq \mathcal{S}_{\text{C}} \,|\,
    &\forall s_i, s_j \in S, s_i \neq s_j, \exists~n, m \in \mathbb{N} \text{ such that } \\
    &  \textit{ParseTrees}(\textit{Eg}(s_i, s_j, n)) \neq \emptyset \lor
       \textit{ParseTrees}(\textit{Eg}(s_j, s_i, m)) \neq \emptyset,~S \text{ is maximal}
\}
\end{align*}
Each set of symbols in $\mathcal{S}_{\text{E}}$ are the maximal subsets of 
$\mathcal{S}_{\text{C}}$ such that for all pairs $s_i$ and $s_j$ in 
the set, the grammar (describing language $\mathcal{L}_{\text{g}}$) 
permits trees $\textit{Eg}(s_i, s_j, n)$ \emph{or} 
$\textit{Eg}(s_j, s_i, m)$ for some child positions $n$ and $m$.

For each of these sets, \greta generates precedence-related tree
examples for the pairs $\textit{Eg}(s_i, s_j, n)$ and $\textit{Eg}(s_j,
s_i, m)$ if they can be parsed in both orders, to obtain a total order
of the symbols in the set. 
It also generates associativity-related tree examples for those in
conflict. It then understands the desired restrictions by interaction
with the user.

We have now laid out notions of tree examples, associativity, precedence,
and excluded parse trees. Our formalism for tree automata and 
conversion from CFGs to TAs is relatively standard, and can be found
in the supplementary material. We now turn to \tool's algorithm 
for TA learning.

\subsubsection{Base Precedence Order}
\label{subsec:base-precedence-order}

\autoref{algo:oaop-learner} starts by computing a set of existing
precedence orders for all the symbols in $\mathcal{F}$ in
$\mathcal{G}$, referred to as the \emph{base precedence order}
$O_{\text{bp}}$, in the following way:
\begin{itemize}
  {\item First, a \emph{level} (or the distance from the start nonterminal) 
  for each nonterminal $e \in V$ is
  \begin{itemize}
    \item $d(e) = 0$ if $e$ is a start nonterminal of $\mathcal{G}$.
    \item Otherwise, $d(e)$ is the smallest $n$ such that there exists
    a sequence of nonterminals $nt_0, \ldots, nt_{n}$ where $nt_{n} = e$,
    $nt_0$ is the start nonterminal, and for each $nt_i, nt_{i+1}$, there
    exists a production $nt_i \rightarrow \beta$ in $P$ such that $\beta$ 
    contains $nt_{i+1}$.
  \end{itemize}}
\item Next, an \emph{order} of each symbol $s$ in $\mathcal{F}$ is
  determined using the order function
  $\widehat{o} \colon \mathcal{F} \mapsto \mathbb{Z}$, defined
  $\widehat{o}(s) \eqdef d(\textit{Lhs}(\textit{Prod}(s)))$, where
  $\textit{Prod}(s)$ is the unique nonterminal of a ranked symbol,
  and $\textit{Lhs}$ denotes the left-hand side nonterminal of a
  production.
\item Lastly, $O_{\text{bp}}$ is constructed as a set
  $\{ (s, \widehat{o} (s)) \,|\, s \in \mathcal{F} \setminus 
  \mathcal{F}_{\text{tr}}\}$, where $\mathcal{F}_{\text{tr}}$ is 
  the set of trivial symbols defined below.
\end{itemize}
One can think of an order of a symbol $s$ as the shortest
distance to reach $s$ from the start nonterminal. In the absence of
cycles, a symbol at a higher order always appears deeper in the parse
tree than a symbol at a lower order. Correctly stratifying symbol
order based on user preference allows us to correct precedence 
ambiguities between different symbols and associativity ambiguities
between a symbol and itself. To account for cycles in order, we will need 
to reintroduce these cycles later in the TA learning process
described in \autoref{sec:tapref}.
%
Now, looking at the running example $\mathcal{G}$
from~\autoref{fig:cfg} with its start nonterminal \code{stmt}, we
have the following levels for the nonterminals: $0$~for \code{stmt},
1~for both \code{decl} and \code{expr}, and 2~for \code{ident}.
Based on this, we can determine orders for all the symbols. Looking at
the symbol \code{if6} and its production \code{stmt -> IF expr THEN
  stmt ELSE stmt}, for instance, the order of \code{if6} is computed
$d$(\code{stmt}) = $0$.
%
Similarly, we can compute orders for all the symbols, producing
$O_{\text{bp}}$ in~\autoref{subsec:restrictions-as-ta}.

The set of trivial symbols $\mathcal{F}_{\text{tr}}$ is a set of
$1$-arity symbols defined as follows:
\begin{align*}
\mathcal{F}_{\text{tr}} \eqdef \set{ 
~s
~\left|~
\begin{array}{l}
s \in \mathcal{F} \land \textit{Rank}(s) = 1 \land \kronecker_{s} \eqdef
  \alpha \leftarrow_{s} \beta \in \Delta \text{ where } \beta \in \Sigma~\land
\\[2pt]
\forall \, \kronecker_{s'} \in \Delta \,\textit{s.t.}\, \kronecker_{s'} = \alpha \leftarrow_{s'} \beta', \beta' \in \Sigma
\end{array}
\right.
}
\end{align*}
$\mathcal{F}_{\text{tr}}$ are symbols whose left hand side nonterminals
only lead to a single terminal symbol. Because we know that symbols in 
$\mathcal{F}_{\text{tr}}$ will not be involved in any ambiguities \wrt 
associativity or precedence order, $\mathcal{O}_{\text{bp}}$ does not 
include symbols in $\mathcal{F}_{\text{tr}}$ and the productions involving 
trivial symbols remain intact in the generated TA. Note however, that this
separate handling of trivial symbols is only an optimisation and does not
affect the correctness of the learning algorithm. It is also the only such
optimisation in \autoref{algo:oaop-learner} and \autoref{algo:ta-learner}.

The construction of the base order $O_{\text{bp}}$ is crucial for
obtaining the relative precedence orders of symbols not involved
in the tree examples. 
If a TA were to be synthesised based on tree examples alone, the
learned TA would feature only those symbols associated with the
ambiguities, which is typically a relatively small subset of the
original symbol alphabet.
At the same time, we believe that generating examples that would
collectively contain all symbols of the input grammar would be
detrimental for the usability of our tool.
This is why the examples provided by \tool only contains a subset of
the grammar's symbols, leaving aside the ones that were not involved
in the ambiguities from the original grammar. 
%
%
Preserving the precedence orders $O_{\text{bp}}$ of the original grammar
allows \tool to learn a TA involving \emph{all} precedence and associativity
conflicts efficiently, without compromising soundness \wrt its
correctness statement in \autoref{sec:soundness-ta-learning}.

\subsubsection{Computing Associativity and Precedence Order}

\begin{algorithm}[t]
\setstretch{0.95}
\SetKwInput{KwInput}{Input}
\SetKwInput{KwOutput}{Output}
\SetKwFunction{FMain}{Main}
\DontPrintSemicolon
\SetNoFillComment
\footnotesize

  \KwInput{tree examples $T^{-}$, order to ordered symbols map $M_{\text{to}}$, input CFG $\mathcal{G}$}
  \KwOutput{associativity $O_{\text{a}}$, precedence order $O_{\text{p}}$}
  
  $O_{\text{bp}} \leftarrow \text{ base precedence order of } \mathcal{G}$\;
  $O_{\text{a}}, O_{\text{p}} \leftarrow \{\};\, O_{\text{tmp}} \leftarrow O_{\text{bp}}$ 
  \tcp{temporary set initialised with all elements in $O_{\text{bp}}$}
  \For {$t \in T^{-}$}{
    \If {$t_T = t_B$} {
      $O_{\text{a}} \leftarrow O_{\text{a}} \cup \{(t_T, t_{idx})\}$
    }
  }
  \For { \text{descending} $(o, G)$ in $M_{\text{to}}$ } {
    $\text{size} \gets \text{maxSize}(G)$\\
    $S = \text{symbols}(\text{ofOrder}(O_{\text{tmp}}, o)) \setminus \bigcup\limits_{g \in G} g $\\
    $O_{\text{tmp}} \gets \text{pushN}(O_\text{tmp}, o + 1, size - 1) \setminus \text{ofOrder}(O_{\text{tmp}}, o)$\\
    \For {i in [$0$, size)} {
      $\text{ithSymbols} \gets \text{getAtIndex}(G, i)$\\
      $O_{\text{tmp}} \gets O_{\text{tmp}} \cup \text{withOrder}(S \cup \text{ithSymbols}, o + i)$\\
      \If {$i = size - 1 \land \exists~s \in \text{ithSymbols}, (s,\_) \in O_{\text{a}}$} {
        $O_{\text{tmp}} \gets \text{pushN}(O_\text{tmp}, o + i + 1, 1)$\\
        $O_{\text{tmp}} \gets O_{\text{tmp}} \cup \text{withOrder}(S, o + i + 1)$
  }
    }
  }
  $O_{\text{p}} \leftarrow O_{\text{tmp}}$ \;
  \Return{} $(O_{\text{a}}, O_{\text{p}})$
\caption{\textsc{LearnOaOp}: learning associativity and precedence orderings.}
\label{algo:oaop-learner}
\end{algorithm}

As illustrated in~\autoref{algo:oaop-learner}, the
{\small\textsc{LearnOaOp}} procedure takes tree examples \emph{not}
chosen by the user $T^{-}$, the input CFG $\mathcal{G}$ and a map
$M_{\text{to}}$ from order to a set of ordered sets of symbols from
$\mathcal{S}_{\text{E}}$ that are at that order, from lowest to highest
precedence, inferred from $T^{-}$. \footnote{ Since all pairs of
symbols in the sets of $\mathcal{S}_{\text{E}}$ can be in precedence
conflict (\ie parsed in either order), they must have the same 
order (allowing us to use it to index $M_{\text{to}}$) or are at adjacent
orders. If they are at adjacent orders, it is sound to merge the symbols at 
the two orders into one. } 
Intuitively, we care about negative tree examples in particular because 
the goal of the learning algorithm is to \emph{exclude} undesirable patterns
appearing anywhere in the parse tree, as opposed to preserving desired
patterns. This is made clear in the description of the algorithm's 
correctness in \autoref{sec:soundness-ta-learning}.
The algorithm then returns restrictions on associativity $O_{\text{a}}$ 
and precedence order $O_{\text{p}}$.
%
%

The goal of \autoref{algo:oaop-learner} is to deduce a set
$O_{\text{a}}$ of associativity restrictions (illegal child positions)
and a set $O_{\text{p}}$ which enforces a hierarchy of precedence
orders among conflicting symbols, while preserving other pairwise
orders from $O_{\text{bp}}$. In section \ref{sec:tapref}, each order
in $O_{\text{p}}$ will correspond to states in the learned TA, where
some state $e_i$ corresponding to order $i$ can be reached by epsilon
transitions from state $e_j$ for $j > i$. 
With $O_{\text{bp}}$, {\small\textsc{LearnOaOp}} collects restrictions 
related to associativity and precedence order, respectively in 
$O_{\text{a}}$ and $O_{\text{p}}$.

First, each tree is examined to check if it is an associativity
or precedence order-related example. If it is an associativity-related
example, then we store in $O_{\text{a}}$ a pair of the symbol and the 
position of the nested symbol we want to disallow.
For example, in~\autoref{fig:user-interact-a}, the tree \emph{not}
selected contains symbol \code{plus3} as a right child (position 1).
Then a pair $($\code{plus3}$,\, 1)$ is added to the set
$O_{\text{a}}$.
Combined with the tree selected in~\autoref{fig:user-interact-b}, the
user interaction produces a set $O_{\text{a}} = \{ ($\code{plus3}$,\,
1), ($\code{star3}$,\, 1)\}$.

Next, the algorithm iterates over $M_{\text{to}}$ in descending order
of the symbol orders. For each order $o$ and corresponding set of
ordered sets of symbols $G$, it first computes the size of the largest
set in $G$. It then stores in $S$ the set of symbols at order $o$ that
are not involved in any conflict. Then, it removes the symbols at
order $o$ from $O_{\text{tmp}}$ and calls \code{pushN}, which increments the
orders of all symbols of order $\geq o + 1$ by $size - 1$, to make
room for the new symbols.

Then, for each order $o + i$ from $o$ to $o + size - 1$, it inserts
the $i$-th set of symbols in $G$, and the auxiliary set of symbols $S$
that are not involved in conflicts. \code{getAtIndex} returns the
symbols at index $i$ in every set in $G$, if it exists, and
\code{withOrder} assigns the specified order to the provided symbols.

As an example, a map $M_{\text{to}} = \{ 0 \rightarrow (\text{\code{star3}, \code{plus3}}) \}$ 
would update $O_{\text{tmp}}$ as follows:

\begin{lstlisting}[basicstyle=\linespread{1.0}\ttfamily\footnotesize,numbers=none]
                   {(*@\,@*)((SEMI,2),0), ((PLUS,3),0), ((STAR,3),0)(*@\,@*)}
                                        _downarrow
             {(*@\,@*)((SEMI,2),0), ((STAR,3),0), ((SEMI,2),1), ((PLUS,3),1)(*@\,@*)}
\end{lstlisting}

\noindent
Here, the set $S = \{ $\text{\code{semi2}}$ \}$,
\code{getAtIndex}($G, 0$) is \code{star3} and
\code{getAtIndex}($G, 1$) is \code{plus3}.

This procedure preserves precedence order relations between pairs 
of symbols that are not involved in the tree examples.
Finally, for all associativity-related symbols in $O_{\text{a}}$, 
we want that if it is present at level $o$, then the non-conflicting 
symbols at that level are also present at level $o+1$, for later 
TA construction. This is automatically ensured for orders $o$ to 
$o + size - 2$ in the algorithm, but might require an additional 
push and insert operation for order $o + size - 1$.
In addition, \textsc{LearnOaOp} repeatedly shifts and reassigns orders
while reinserting conflicting symbols, resulting in
$O(|\mathcal{F}|^{2})$ time and $O(|\mathcal{F}|)$ additional space
where $|\mathcal{F}|$ refers to the size of $\mathcal{F}$.
In the above where \code{plus3} and \code{star3} are involved in
associativity conflicts, the final set $O_p$ is:

\begin{lstlisting}[basicstyle=\linespread{1.0}\ttfamily\footnotesize,numbers=none]
       {(*@\,@*)((SEMI,2),0), ((STAR,3),0), ((SEMI,2),1), ((PLUS,3),1), ((SEMI,2),2)(*@\,@*)}
\end{lstlisting}
\vspace*{-2pt}

%
%

\begin{algorithm}[!t]
\SetKwInput{KwInput}{Input}
\SetKwInput{KwOutput}{Output}
\SetKwFunction{FMain}{Main}
\DontPrintSemicolon
\SetNoFillComment
\footnotesize

  \KwInput{$O_{\text{a}}$ associativity set, $O_{\text{p}}$ precedence order set, $\mathcal{G} = (V, \Sigma, S, P)$ input CFG}
  \KwOutput{\capcal{A} $= (Q, \mathcal{F}, \Sigma, Q_{f}, \Delta)$ a tree automaton}

  $m \leftarrow \text{max order of } O_{\text{p}}; \,\, Q \leftarrow \{ e_{0},
  \,\ldots,\, e_{m} \}; \,\, Q_{f} \leftarrow \{ e_{0} \}$ \\ 

  $\mathcal{F} \leftarrow \text{ranked symbols of } \mathcal{G};
  \,\,\mathcal{F}_{\text{tr}} \leftarrow \text{trivial symbols of } \mathcal{G}$ \\ 
  $\kronecker_{\mathcal{F}} \leftarrow \kronecker\text{-generator } w.r.t.
  \text{productions } P$ \\

  \tcc{Transitions for non-trivial symbols}
  \For {$(s, i) \in O_{\text{p}}$} {
    \If {$\exists p, (s, p) \in O_{\text{a}}$} { 
      \tcp{$\,\kronecker_{s} \eqdef e_{i} \leftarrow_{s} \,\, \dots e_{i+1} \dots$
      ($e_{i+1}$ at position $p$ in the RHS, with $n$ leading and $m$ trailing
      nonterminals respectively)}
      $\Delta \leftarrow \Delta \cup \{
      \kronecker_{\mathcal{F}}(e_{i}, [e_{i};\dots e_{i+1};\dots e_{i}], s ) \}$ 
      \tcp{$n$ leading and $m$ trailing $e_i$'s}
    } \Else { 
      $\Delta \leftarrow \Delta \cup \{ \kronecker_{\mathcal{F}}(e_{i}, \widebar{e_{i}}, s ) \}$
    }
  }
  \tcc{Transitions for trivial symbols}
  \For {$s' \in \mathcal{F}_{\text{tr}} $} {
    $Q \leftarrow Q \cup \{ e_{s'} \} ; \,\, \Delta \leftarrow \Delta \cup \{
    \kronecker_{\mathcal{F}}(e_{s'}, [], s' ) \}$
    \tcp{$i.e., \,\kronecker_{s'} \eqdef e_{s'} \leftarrow_{s'} \alpha$}
  }
  \tcc{Transitions for connecting ordered states}
  \For {$i \in [0..m-1]$} {
    $\Delta \leftarrow \Delta \cup \{\kronecker_{\mathcal{F}}(e_{i}, [e_{i+1}],
    (\epsilon,1) )\}$
    \tcp{$i.e., \,\kronecker_{(\epsilon,1)} \eqdef e_{i}
    \leftarrow_{(\epsilon,\, 1)} e_{i+1}$}
  }
  \tcc{Handle cycles in order}
  \For {$(sl, ol), (sh, oh) \in \text{HighToLow}(G, O_{\text{p}})$} {
    $\Delta \leftarrow \Delta \cup \{ \kronecker_{\mathcal{F}}(e_{oh}, \widebar{e_{ol}}, sh) \}$
  }
  \Return{} $(Q,\, \mathcal{F},\, \Sigma,\, Q_{f},\, \Delta)$ \\
  \caption{\textsc{GenTA}: generating a TA from associativity,
    precedence, and the input grammar.}
\label{algo:ta-learner}
\end{algorithm}

\subsubsection{A Tree Automaton for Inferred Preferences}
\label{sec:tapref}

Based on the preferences specified by the associativity set
$O_{\text{a}}$, the precedence order set $O_{\text{p}}$, and the input
CFG $\mathcal{G}$, \autoref{algo:ta-learner} constructs a TA
$\mathcal{A} = (Q, \mathcal{F}, \Sigma, Q_{f}, \Delta)$ that encodes the given
parsing preferences over $\mathcal{F}$.
%
The set of its states $Q$ is initially populated with states associated
with a range of levels, from $0$ to the maximum order $m$ in
$O_{\text{p}}$, representing each level in the hierarchy of orders: 
that is, $Q \eqdef \{ e_{0}, \ldots, e_{m} \}$.
These \emph{ordered} states are used to generate transitions to encode
precedence relations among symbols specified in $O_{\text{p}}$, with
$Q_{f} \eqdef \set{e_{0}}$.
%
%

Next, we have a $\kronecker$-\emph{generator}
$\kronecker_{\mathcal{F}} \colon Q \times Q^{*} \times \mathcal{F} \mapsto \Delta$.
%
The function $\kronecker_{\mathcal{F}}$ takes the left hand side state,
right hand side list of states, and a symbol in $\mathcal{F}$ to 
produce transition with the symbols in the right positions.
In other words, given a transition structure for a symbol in
$\mathcal{F}$ obtained from the productions $P$ in $G$,
$\kronecker_{\mathcal{F}}$ replaces all the nonterminals with the left
and right hand states provided, while keeping all the terminal 
symbols.
We write $\widebar{e_i}$ to denote an appropriately lengthed list 
(dictated by the number of nonterminals in the production) 
filled with the state $e_i$.
%
%
For example, given a symbol \code{if4}, when we provide
$\widebar{e_1}$ of length $2$ as an input list of states,
$\kronecker_{\mathcal{F}}(e_1, \,\widebar{e_1}, \,$\code{if4}$)$
returns \code{e_1 <-_if4- IF e_1 THEN e_1} which is a result of
replacing all the nonterminals from the original production \code{stmt
-> IF expr THEN stmt} with a state \code{e_1} while keeping all the
terminals \code{IF} and \code{THEN} intact.\footnote{In this case, we
replace all the nonterminals---\texttt{stmt}, \texttt{expr}, and
\texttt{stmt}---with \code{e_1} as it is listed as the only new state
to replace the old states and there is no nonterminal associated with
the symbols in $\mathcal{F}_{\text{tr}}$. Note that if the input list
is of length $> 1$, we replace each old state (excluding the states
associated with trivial symbols) with the new states in the input list
starting from the head element consecutively. For example, given
structure of \texttt{(PLUS, 3)}-transition \texttt{expr}
\code{<-_plus3-} \texttt{expr PLUS expr}, applying the
$\kronecker$-generator $\kronecker_{\mathcal{F}}(e_2,[ e_2; e_3 ],
$\texttt{(PLUS, 3)}$)$ produces \code{e_2} \code{<-_plus3-} \code{e_2}
\texttt{PLUS} \code{e_3}. }
%
%
This procedure of retrieving a template (predefined transition rule)
from a map and replacing certain (nonterminal) placeholders with
provided contents (new nonterminals) is comparable to template-based 
code generation~\cite{Syriani:CLSS2018}.

For each $(s, i) \in O_{\text{p}}$, if there exists a pair $(s, p) 
\in O_{\text{a}}$ (i.e. $s$ is involved in an associativity conflict),
then a transition rule $\kronecker_{\mathcal{F}}(e_{i}, [e_{i}; \dots 
e_{i+1}; \dots e_{i}], s)$ is added to $\Delta$, where there the list
of right hand side states include $n$ leading $e_{i}$'s then $e_{i+1}$
followed by $m$ trailing $e_{i}$'s. This corresponds to $e_{i+1}$
appearing at a postition $p$ in the right hand side of the production
with $n$ leading nonterminals. This ensures that child symbol does 
\emph{not} appear at the $p$-th position.
For example, $($\code{plus3}$,\, 1)$ in $O_{\text{a}}$ along
with $($\code{plus3}$,\, 2)$ in $O_{\text{p}}$ adds a rule
$\kronecker_{\mathcal{F}}(e_{2}, [e_{2}; e_{3}], s) \eqdef$ \code{e_2
<-_plus3- e_2 *plus* e_3} to $\Delta$. 

When a symbol $s$ (at order $i$) is not contained in $O_{\text{a}}$
then we simply add $\kronecker_{\mathcal{F}}(e_{i}, \widebar{e_i},
s)$. For example, $($\code{if4}$,\,1)$ in $O_{\text{p}}$ results in
$\kronecker_{\mathcal{F}}(e_{1}, \widebar{e_1}, \,$\code{if4}$)
\eqdef$ \code{e_1 <-_if4- IF e_1 THEN e_1} being added to $\Delta$.

Next, for each symbol $s'$ in $\mathcal{F}_{\text{tr}}$ and its
associated state $e'$, we add $\kronecker_{\mathcal{F}}(e', [],
s')$ to $\Delta$ and $e'$ to $Q$, thus allowing a transition such as
\code{ident <-_ident1- IDENT} to be included in $\Delta$.
%
Following the addition of rules for the trivial symbols, we connect
the states of different levels, by generating \code{eps1}-transitions
\emph{consecutively}, from $e_0$ to $e_1$, $e_1$ to $e_2$, and so on
until $e_{m-1}$ to $e_{m}$.

\begin{wrapfigure}[10]{r}{0.43\textwidth}
\vspace{-10pt}
  \centering
  \includegraphics[width=0.42\textwidth]{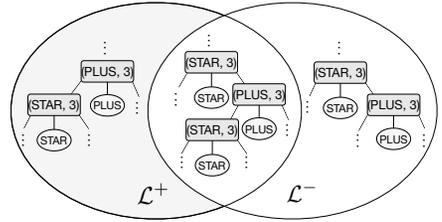}
  \caption{Trees in $\mathcal{L}^{+}$ and $\mathcal{L}^{-}$.}
  \label{fig:lang-diagram}
\end{wrapfigure}

To handle cycles in the order of symbols correctly, we also need to
add transitions that allow for these cycles in the learned automata.
Consider symbols $s_l$ and $s_h$ where $s_h$ is at a higher order in
$O_{\text{bp}}$ than $s_l$, with productions $e_l \rightarrow_{s_l}
\beta_l$ and $e_h \rightarrow_{s_h} \beta_h$ in $\mathcal{G}$ such that 
$e_l \in \beta_h$---i.e., $s_l$ can appear deeper in a parse tree than
 $s_h$. Intuitively, such circularity in symbol order needs to be
reintroduced in the automaton. The function \textit{HighToLow}($G$,
$O_\text{p}$) identifies such pairs and returns $(s_h, o_h), (s_l,
o_l)$ where $o_h$ is the highest order of $s_h$ and $o_l$ is the 
lowest order of $s_l$ in $O_\text{p}$. For each such pair, we add 
the transition $\kronecker_{\mathcal{F}} (e_{o_h}, \widebar{e_{o_l}}, 
s_h)$ connecting them.


Lastly, the transition rules in $\Delta$ along with $Q$,
$\mathcal{F}$, and $Q_{f}$ result in a TA, as shown
in~\autoref{fig:ta-restrictions}.
Given that $|\mathcal{F}|$ is the total number of symbols,
\textsc{GenTA} has time and space complexity of
$O(|\mathcal{F}|)$. 
Details of the complexity analysis can be found in the supplementary
material.

\subsubsection{Soundness of Tree Automata Learning}
\label{sec:soundness-ta-learning}

\autoref{algo:ta-learner} is sound in the following sense: a tree
automaton it constructs faithfully encodes associativity and
precedence order, and does not lose unrelated information from 
the input grammar~$\mathcal{G}$.
Given trees $T^{+}$ selected by the user, and trees $T^{-}$ 
not selected by the user, we define the tree languages:
$\mathcal{L}^{+} \eqdef \bigcup\limits_{t \in T^{+}} \text{\textit{ParseTrees}}(t)$ and
$\mathcal{L}^{-} \eqdef \bigcup\limits_{t' \in T^{-}} P^{-}(t')$.
$L^{+}$ is the set of all parse trees described by the selected tree
examples, and $L^{-}$ is the set of trees that are excluded by \tool,
which strictly enforces symbol hierarchy as described in \autoref{subsec:learner}.
These sets can have an overlap, as illustrated in \autoref{fig:lang-diagram}.
We want our learned automaton to reject all trees described by $L^{-}$, which 
will also exclude this intersected region of $L^{+}$.

With these definitions in hand, we are ready to state our main
soundness result.


%


\begin{theorem}[Soundness of \textsc{GenTA}]
  \label{thm:soundness-learner}
  Let $\mathcal{A}_{\text{r}}$ be the finite TA returned
  by~\autoref{algo:ta-learner} and $\mathcal{A}_{\text{g}}$ be the TA 
  derived from CFG $\mathcal{G}$, $T^{-}$ be the negative tree examples.
  Further, $\mathcal{L}_{\text{r}} = L(\mathcal{A}_{\text{r}})$, 
  $\mathcal{L}_{\text{g}} = L(\mathcal{A}_{\text{g}})$,
  and $\mathcal{L}^{-} = \bigcup\limits_{t \in T^{-}} P^{-}(t)$.
  %
  Then, $\mathcal{L}_r \supseteq \mathcal{L}_{\text{g}} \setminus
  \mathcal{L}^-$, and $\mathcal{L}_r \cap \mathcal{L}^- = \emptyset$.
\end{theorem}

\begin{proof}
  Provided in the supplementary material.
\end{proof}

%
The first statement asserts that $\mathcal{A}_{\text{r}}$ does not
lose any unnecessary parse trees in $\mathcal{L}_{\text{g}}$ that are 
not excluded by the negative examples in $T^{-}$ (which includes the
parse trees selected by the user).
The second statement asserts that $\mathcal{A}_{\text{r}}$ does not
accept any tree that belongs to $\mathcal{L}^{-}$, that have been 
discarded by the user.

%

%
%

%
We conclude by noting that the CFG translated from the resulted automaton
is \emph{not} readily usable as a repair solution of the original
grammar. 
That is because the TA is generated to be permissive enough to include
all the original grammar's parse trees that do not have to do with
conflicts, but this makes the TA accept even the trees that should not
be allowed.
For example, a tree for \code{*if* 1 THEN 2 PLUS 3} is accepted by
$\mathcal{A}_{\text{r}}$, whereas it is not accepted by
$\mathcal{A}_{\text{g}}$.
This leads us to discuss our next contribution, an algorithm for
intersection of $\mathcal{A}_{\text{r}}$ and $\mathcal{A}_{\text{g}}$.

\begin{figure}[t!]
\begin{minipage}{0.28\textwidth}
\begin{lstlisting}[style=HL,numbers=none]
ss:
| s ss 
|  /*  empty  */ 
s: 
| dcl SEMI  
| id EQ x SEMI
| IF LPAREN x RPAREN s
| IF LPAREN x RPAREN s ELSE s
| RETURN x SEMI
| WHILE LPAREN x RPAREN s
| LBRACE ss RBRACE
e:
| e PLUS e
| e DASH e  
| e STAR e
| id 
| in(*@t@*)  
| LPAREN e RPAREN

id: IDENT
in(*@t@*): INT 
dcl: TINT id EQ e
\end{lstlisting}
\vspace{32pt}
\subcaption{{Original grammar}}
\label{fig:grammar-before}
\end{minipage}
\begin{minipage}{0.33\textwidth}
\begin{lstlisting}[style=HL,numbers=none]
ss:
| s1 ss
| /* empty */ 
s1:
| IF LPAREN e1 RPAREN s1
| s2 
s2:
| dcl SEMI 
| id EQ e1 SEMI
| WHILE LPAREN e1 RPAREN s2 
| RETURN e1 SEMI
| LBRACE ss RBRACE 
| IF LPAREN e1 RPAREN s2 ELSE s2
e1:
| e1 PLUS e2 | e2
e2:
| e2 DASH e3 | e3 
e3:
| e3 STAR e4 | exp4
exp4:
| id
| in(*@t@*)
| LPAREN e1 RPAREN 
id: IDENT
in(*@t@*): INT
dcl: TINT id EQ e1
\end{lstlisting}
\subcaption{{Repair via~\tool}}
\label{fig:grammar-greta}
\vspace{0.5\baselineskip}
\end{minipage}
\begin{minipage}{0.37\textwidth}
\begin{lstlisting}[basicstyle=\linespread{1.0}\ttfamily\tiny,numbers=none]
(ss,e1): 
| (s,e1) (ss,e1)
(s,e1):
| (i(*@f@*),i(*@f@*)) (lparen,lparen) (e,e1) 
  (rparen,rparen) (s,e1)
(e2,e2): 
| (in(*@t@*),e3)
| (id,e3)
(ss,e1): 
| (ε,ε)
| (ε,e2)
(e,e1): 
| (e2,ε)
| (e2,e2)
| (e2,e4)
| (e2,e3)
(e,e2):
| (e,e2) (dash,plu(*@s@*)) (e,e2)
| (e,e2) (dash,dash) (e,e2)
| (e,retur(*@n@*)) (dash,e2) (e,semi)
| (e,lparen) (dash,e1) (e,rparen)
| (e,e2) (dash,dash) (e,e2)
| (e,lbrace) (dash,e1) (e,rbrace)
(s,e2):
| (retur(*@n@*),e2) (exp,plu(*@s@*)) (semi,e2)
| (retur(*@n@*),e2) (exp,dash) (semi,e2)
| (retur(*@n@*),retur(*@n@*)) (e,e2) (semi,semi)
| (retur(*@n@*),lparen) (e,e1) (semi,rparen)
: 
\end{lstlisting}
\subcaption{{Repair via classical TA intersection}}
\label{fig:grammar-txtbk}
\vspace{0.5\baselineskip}
\end{minipage}
\setlength\abovecaptionskip{5pt}
\setlength\belowcaptionskip{-15pt}
\caption{Disambiguated grammar returned by \tool vs. simple cross products.}
\label{fig:idiomatic-grammar}
\end{figure}


\subsection{Intersecting Tree Automata for Context-Free Grammars}
\label{subsec:intersect}

%
%
Our intersection algorithm ultimately computes the component-wise 
cross-products of sets of states, accepting states, and transition 
rules~\cite{Comon-et-al:tata08}. 
In the standard definition of finite tree automata (FTA) intersection,
this cross-product construction is applied uniformly to all states and
transitions, combining transitions solely based on matching arity of
ranked symbols and enumerating all resulting state tuples, regardless
of whether they can contribute to accepting run.  
%
%
\autoref{algo:intersect} departs from the textbook construction in two
key aspects that are essential in our grammar-based setting: (1) it
gives special treatment to transitions involving terminals (\ie,
producing the product of transitions not only based on matching arity,
but also the precise sequence of terminals and nonterminals on the
right-hand side), and (2) it incorporates reachability-based
optimisations directly into the construction. 
%
%
That is because (1) the standard definition does not account for
transitions involving terminals, producing a more complex grammar due
to states mapped from all the existing terminals. This helps us to
produce a more idiomatic grammar, as shown by the grammar
in~\autoref{fig:grammar-greta} as compared to the results taken from
the textbook definition of cross-products
in~\autoref{fig:grammar-txtbk}.
We introduced (2) because simply taking cross-products of each
component yields a number of states and transitions that cannot 
reach the final states, which we can avoid computing entirely,
by leveraging reachability analysis.

\begin{algorithm}[!t]
\setstretch{0.95}
\SetKwInput{KwInput}{Input}
\SetKwInput{KwOutput}{Output}
\SetKwProg{Fn}{Function}{:}{}
\DontPrintSemicolon
\SetNoFillComment
\footnotesize

  \KwInput{$\mathcal{A}_{\text{g}} = (Q_{\text{g}}, \mathcal{F}, \Sigma, Q_{\text{gf}}, \Delta_{\text{g}})$ and 
          $\mathcal{A}_{\text{r}} = (Q_{\text{r}}, \mathcal{F}, \Sigma, Q_{\text{rf}}, \Delta_{\text{r}})$ tree automata}
  \KwOutput{\capcal{A} $= (Q, \mathcal{F}, \Sigma, Q_{\text{f}}, \Delta)$ a tree automaton}
  $Q_{\text{f}} \leftarrow Q_{\text{gf}} \times Q_{\text{rf}}$ \\

  \tcc{Learn $\Delta$ \wrt reachable states}
  $Q, Q_{\text{tmp}}  \leftarrow Q_{\text{f}}$ \\ 
  \While{$Q_{\text{tmp}} \neq \emptyset$} {
    \For {$(e_{\text{g}}, e_{\text{r}}) \in Q_{\text{tmp}}$} {
      $\mathcal{F}_{\text{g}} \leftarrow \text{symbols of } e_{\text{g}} \text{ in } \Delta_{\text{g}}; \,\, \mathcal{F}_{\text{r}} \leftarrow \text{symbols of } e_{\text{r}} \text{ in } \Delta_{\text{r}}$ 
      \tcc{Includes symbols reachable by \code{eps1}-transitions}

      $\mathcal{F}' \leftarrow \mathcal{F}_{\text{g}} \cap \mathcal{F}_{\text{r}}$ \\
      \For {$s \in \mathcal{F}'$} {
          $\kronecker^{g}_{s} \leftarrow \kronecker_{s} \text{ producing } e_{\text{g}} \text{ in } \Delta_{\text{g}}; \,\, \kronecker^{r}_{s} \leftarrow \kronecker_s \text{ producing } e_{\text{r}} \text{ in } \Delta_{\text{r}}$ \\
          $\Delta \leftarrow \Delta \cup \{ \kronecker^{g}_{s} \times \kronecker^{r}_{s} \}$ \\
          $Q' \leftarrow \text{reachable states of } (e_{\text{g}}, e_{\text{r}}) \text{ in } \kronecker^{g}_{s} \times \kronecker^{r}_{s}$ \\
          $Q_{\text{tmp}} \leftarrow Q_{\text{tmp}} \cup \{ q \, | \, q \in Q' \text{ and } q \notin Q \} \backslash \{ (e_{\text{g}}, e_{\text{r}}) \}$; \,
          $Q \leftarrow Q \cup Q'$
           
      }
    }
  }

  \tcc{Remove duplicate states}
  $Q_{\text{dup}} \eqdef \textsc{FindDupStates} (Q, \Delta)$ \\

  \For {$((x_{\text{g}}, x_{\text{r}}), (y_{\text{g}}, y_{\text{r}})) \in Q_{\text{dup}}$} {
    $Q \leftarrow Q \backslash \{ (y_{\text{g}}, y_{\text{r}}) \}$ \\
    $\Delta \leftarrow \Delta \backslash \{ \kronecker \, | \, \kronecker \text{ producing } (y_{\text{g}}, y_{\text{r}}) \}$ \\
    $\text{Replace } (y_{\text{g}}, y_{\text{r}}) \text{ with } (x_{\text{g}}, x_{\text{r}}) \text{ in } \Delta$ 
  }
  
  \tcc{Introduce $\epsilon$-transitions to simplify $\Delta$}
  $L_{q} \leftarrow \text{Ordered list of } Q \text{ whose } | \kronecker | \text { is smallest to largest} $ \\
  \For {$i \in [ 0 \, .. \, |L_{q}| )$} {
    $\Delta_{i} \leftarrow \{ \kronecker \, | \, \kronecker \in \Delta \text { producing i\textsuperscript{th} } Q \text{ in } L_{q} \}$ \\
    \For {$j \in [ i+1 \, .. \, |L_{q}| )$} {
        $\Delta_{j} \leftarrow \{ \kronecker \, | \, \kronecker \in \Delta \text { producing j\textsuperscript{th} } Q \text{ in } L_{q} \}$ \\
        \If {$\text{RHS of } \Delta_{i} \subset \text{RHS of } \Delta_{j}$} {
            $\Delta_{j}' \leftarrow (\Delta_{j} \backslash \Delta_{i}) \cup \{ e_{j} \leftarrow_{(\epsilon, 1)} e_{i} \}$ \\
            $\Delta \leftarrow (\Delta \backslash \Delta_{j}) \cup \Delta_{j}'$
        }
    }
  }
  \Return{} $(Q, \mathcal{F}, \Sigma, Q_{\text{f}}, \Delta)$ 
\caption{\textsc{IntersectTA}}
\label{algo:intersect}
\end{algorithm}


%
\autoref{algo:intersect} summarises the algorithm optimised for
reachability as well as transitions involving terminals, to
efficiently compute the intersection of tree automata. 
Below, we elaborate on some of its components.
First, the algorithm computes a set of accepting states $Q_{\text{f}}$
by taking the cross-product of the sets of accepting states of the
input TAs: $Q_{\text{gf}}$ and $Q_{\text{rf}}$.
Since we are intersecting two TAs, $Q_{\text{f}}$ consists of only
one state (\ie, a pair consisting of states respectively from
$Q_{\text{gf}}$ and $Q_{\text{rf}}$).
Next, the algorithm updates a set of states $Q$ and a temporary list
of states $Q_{\text{tmp}}$ with $Q_{\text{f}}$.
Notice that if we simply take a cross-product of $Q_{\text{g}}$ and
$Q_{\text{r}}$ to populate $Q$, the resulting set $Q$ might contain
states that are not reachable, and, thus, don't have to be included in
$Q$.
Therefore, we add to $Q$ only those states that can reach the final
state in $Q_{\text{f}}$ to result in a TA.
We maintain $Q_{\text{tmp}}$ as a worklist by adding any states we
encounter for the first time and consuming the states whenever we 
generate transitions producing them.
Specifically, for all $(e_{\text{g}}, e_{\text{r}})$ in
$Q_{\text{tmp}}$, we obtain all the transitions corresponding to
$(e_{\text{g}}, e_{\text{r}})$ with labels that are intersection of
the ranked symbols which $e_{\text{g}}$ and $e_{\text{r}}$ can
respectively take in $\Delta_{\text{g}}$ and $\Delta_{\text{r}}$,
possibly with $\epsilon$-transitions.
Then, we update $Q$ and $Q_{\text{tmp}}$ with states that have not
been collected by $Q$ and are \emph{reachable} from---\ie, appearing
on the right-hand sides of the transitions from---$(e_{\text{g}},
e_{\text{r}})$, while we remove $(e_{\text{g}}, e_{\text{r}})$ from
$Q_{\text{tmp}}$.
These steps are repeated to add unseen states to $Q$ and new set of
transitions to $\Delta$ until $Q_{\text{tmp}}$ is empty.

Upon collecting all the \emph{raw} cross-products of the transitions,
we remove any duplicate states, identified via
\textsc{RemoveDupStates} in~\autoref{algo:find-dups}, and their
corresponding transitions, to reduce the TA following the idea of TA
minimisation~\cite{Comon-et-al:tata08}.
%
%
Lastly, we examine $\Delta$ again to determine which (sub)set of
transitions are repeated for states and introduce
\code{eps1}-transitions to simplify the TA further.
%
Moreover, let $|Q_{\text{g}}|, |Q_{\text{r}}|$ and
$|\Delta_{\text{g}}|, |\Delta_{\text{r}}|$ be the numbers of states
and transitions of the two input automata.
Then, \textsc{IntersectTA} costs $O((|Q_{\text{g}}| \cdot
|Q_{\text{r}}|)^{2} \cdot |\Delta_{g}| \cdot |\Delta_{r}|)$ time and
$O((|Q_{\text{g}}| \cdot |Q_{\text{r}}|)^{2} + |\Delta_{g}| \cdot
|\Delta_{r}|)$ space in the worst case.
We include relevant details of complexity analysis in the
supplementary material. 

\begin{theorem}[Correctness of \tool]
  \label{thm:correctness-greta}
  The intersection of automaton $\mathcal{A}_{\text{r}}$
  (\textsc{GenTA}'s result) 
  with $\mathcal{A}_{\text{g}}$ (the automaton derived from CFG $\mathcal{G}$)
  produces a tree automaton recognizing the language $\mathcal{L}_{\text{g}} \setminus \mathcal{L}^{-}$.
\end{theorem}

\begin{proof}
  Follows from Theorem \ref{thm:soundness-learner} and set intersection.
\end{proof}


\begin{algorithm}[t]
  \setstretch{0.95}
  \SetKwInput{KwInput}{Input}
  \SetKwInput{KwOutput}{Output}
  \SetKwProg{Fn}{Function}{:}{}
  \DontPrintSemicolon
  \SetNoFillComment
  \footnotesize

  \KwInput{$Q$ a set of states, $\Delta$ a set of transitions}
  \KwOutput{$Q_{\text{dup}}$ a set of state pairs}
  \For {$e_{i} \in Q$} {
      $\Delta_{i} \leftarrow \text{Transitions to } e_{i} \text{ in } \Delta$ \\
      $\Delta_{i} \leftarrow \text{Replace } e_{i} \text{ with } e_{tmp} \text{ in } \Delta_{i}$ \\
      \For {$e_{j} \in Q \backslash \{ e_{i} \}$} {
          $\Delta_{j} \leftarrow \text{Transitions to } e_{j} \text{ in } \Delta$ \\
          $\Delta_{j} \leftarrow \text{Replace } e_{j} \text{ with } e_{tmp} \text{ in } \Delta_{j}$ \\
          \If {$\Delta_{i} = \Delta_{j}$} {
              $Q_{\text{dup}} \leftarrow Q_{\text{dup}} \cup \{ (e_{i}, e_{j}) \}$
          }        
      }
  }
  \Return{} $Q_{\text{dup}}$
\caption{\textsc{FindDupStates}}
\label{algo:find-dups}
\end{algorithm}

\section{Implementation and Evaluation}
\label{sec:eval}

We implemented the proposed methodology in a tool called \tool.
\tool is written in OCaml and depends on OCaml's standard LR(1) parser
generator \menhir~\cite{Regis:Menhir2016} for identifying ambiguities
in the given grammar.
%
%
We use \menhir, as it can handle more complex grammars than
\texttt{ocamlyacc}, an LALR(1) parser generator, and produces more
detailed and comprehensible error messages for debugging faulty
grammars, making it the parser generator of choice in
OCaml~\cite{Madhavapeddy-Minsky:CUP22}. 
The initial CFG is therefore expressed in \menhir's input format.
%
Then, tree examples are created based on the set of productions that
lead to each conflict identified by \menhir.
%
%

\subsection{Interaction Design}
We designed the user interface of \tool in a way that reduces the
number of choices the user has to make, so \tool would resolve $N$
ambiguities with \emph{at most} $N$ interactions with the user.
We do so by presenting one set of tree alternatives out of the group
of involved tokens (\ie, symbols) per state, where a conflict happens,
according to \menhir.
For example, if a grammar has two precedence order conflicts, one
between \code{if4} and \code{plus3} and another between \code{if4} and
\code{star3}, \menhir reports a conflict at a state reached after
processing the \code{if4}-production where tokens involved are
\code{PLUS} and \code{STAR}.
In this case, \tool presents only one set of trees, each of which
alternate the depths of symbols, \eg, \code{if4} and \code{plus3},
respectively.

%
%
%
We aim to reduce the interaction burden on the user by \emph{not}
requiring them to choose a tree example for every single ambiguity in
the grammar.
%
%
Consider the following scenario where there are three ambiguities: one
between \code{if6} and \code{star3}, another between \code{if4} and
\code{star3}, and third one between \code{if4} and \code{if6}.
Based on our interaction design, we present the two sets of trees to
the user: one showing trees with alternating depths of symbols
\code{if6} and \code{star3} and another one with symbols \code{if4}
and \code{star3}.
Suppose the user has selected precedence relations \code{if6} $>$ \code{star3}
and \code{if4} $<$ \code{star3}.
Based on this, we can infer \code{if4} $<$ \code{if6}. Thus, \tool
ends up resolving three ambiguities with only two tree selections.
Note that if the user has chosen \code{if6} $>$ \code{star3} and 
\code{if4} $>$ \code{star3}, \tool would have required another user 
interaction.
%
Therefore, the number of prompts to the user depends not just on the
conflicts in the input CFG, but also on the set of trees selected by
the user in each interaction.
%

\subsection{Experimental Setup}
A run of \tool involves a series of user prompts, in which the user
specifies their preference between two tree examples. 
\tool then produces a new grammar by applying the methodology
from~\autoref{sec:methodology} and provides it as a new input to
\menhir. 
%
In this process, the user does not have to do or know that \tool
involves any operations on TAs. 
%
When the newly produced CFG still contains ambiguities, there are
subsequent rounds of running \tool involving user interaction until
\tool successfully resolves all the ambiguities \textit{or} left with
ambiguities not addressable by \tool.
We discuss the non-addressable ambiguities in detail
in~\autoref{subsec:scopelimits}. 
Hence, the end-to-end workflow of \greta involves a series of prompts.
Since each prompt provides two tree examples, the cumulative scenarios
become exponential in the number of prompts, which quickly gets
intractable for manual testing.
Our testing framework therefore automates this interaction with
\tname{Expect}~\cite{Libes:Expect1995}, a terminal text interface
automation tool. 
Moreover, since we do not involve real users, we present all
ambiguities in each grammar and test for all possible scenarios.
All our experiments are conducted on a commodity machine running
Ubuntu 22.04, with 16GB RAM and 16 logical CPU cores.

%

\paragraph{Benchmarks}
To evaluate our methodology, we accumulated a variety of grammars
whose sizes vary \wrt numbers of terminals, nonterminals, and
productions, as shown in the first three columns
in~\autoref{tab:grammar-stats}.
%
%
$G0$ is identical to the running example in~\autoref{fig:cfg}. 
$G1$ is similar to $G0$ but more interesting as it allows symbols for
binary operations \code{plus3} and \code{star3} to conflict with
symbols \code{if4} and \code{if6}, making it possible for the grammar
to have more ambiguities.
We obtained the next three grammars $G2$ to $G4$ from course
assignments of a popular class on compiler design.
$G2$ describes a simple boolean language, $G3$ and $G4$ define more
complex languages including while loop and various binary operations
in addition to boolean expressions.
$G5$ and $G6$ are collected from questions posted on StackOverflow. 
$G5$ is a simple grammar written for logical
expressions,\footnote{\url{https://stackoverflow.com/questions/910445/issue-resolving-a-shift-reduce-conflict-in-my-grammar}.
} whereas $G6$ describes a language for expressing constraints with
inequalities and binary
operations.\footnote{\url{https://stackoverflow.com/questions/4588397/fixing-lemon-parsing-confilcts?rq=3}.
}
$G1$ to $G6$ comes in 2 or 3 variants, $GN$a to $GN$b (or $GN$c),
containing different number/type of ambiguities to see how \tool
performs in terms of accuracy and speed with an increasing number of
ambiguities.

We compiled practical grammars of real-world languages: $G7$, $G8$,
and $G9$. $G7$ is the Michelson grammar that is used for specifying
smart contracts on the Tezos blockchain, while $G8$ is the grammar of
Kaitai, a declarative language for describing binary structures in
Tezos.\footnote{\url{https://gitlab.com/tezos/tezos/-/tree/master/client-libs/kaitai-ocaml}}
$G7$ grammar was obtained from a publicly available subset of the Michelson grammar,\footnote{\url{https://github.com/aigarashi/ocaml_of_michelson}} combined with more constructs from the Michelson reference.\footnote{\url{https://tezos.gitlab.io/michelson-reference/}}
Lastly, $G9$ is a subset of the SQL language, a standard language used
for accessing and manipulating databases.

\begin{table}[!t]
\scriptsize
\centering
\begin{adjustbox}{max width=\textwidth}
\begin{tabular}{c@{\quad}cccccccccccccc}
\toprule
 & {$|\Sigma|$} & {$|V|$} & {$|P|$}
 & {A$_{\textcolor{gray}{\scriptsize{p},{a}}}$}
 & {$\Delta$} & {Conv} & {Learn}
 & {$I^{\text{def}}_{\textcolor{gray}{\%}}$}
 & {$I^{\text{1}}_{\textcolor{gray}{\%}}$}
 & {$I^{\text{2}}_{\textcolor{gray}{\%}}$}
 & {$I^{\text{3}}_{\textcolor{gray}{\%}}$}
 & {$I^{\text{123}}_{\textcolor{gray}{\%}}$}
 & {\tikz[remember picture]\node(vline-top){};$T_{\text{man}}$}
 & {$\Delta_{\text{man}}$} \\
\midrule
$G0$ & 13 & 5 & 10
& $\AmbPA{5}{2}{2}$
& 4 & 0.06 & 1.68
& $0.62_{\textcolor{gray}{50}}$ & $3.06_{\textcolor{gray}{50}}$ & $0.91_{\textcolor{gray}{0}}$
& $0.50_{\textcolor{gray}{20}}$ & $2.21_{\textcolor{gray}{0}}$ & 
\:\, 265 & 9.5 \\
\midrule
$G1$a & 11 & 4 & 10
& $\AmbPA{4}{3}{1}$
& 4 & 0.06 & 0.75
& $0.84_{\textcolor{gray}{46}}$ & $3.53_{\textcolor{gray}{46}}$ & $1.36_{\textcolor{gray}{0}}$
& $0.67_{\textcolor{gray}{15}}$ & $3.02_{\textcolor{gray}{0}}$ & 
\:\, 163 & 6 \\
$G1$b & 11 & 4 & 10 
& $\AmbPA{7}{6}{1}$
& 7 & 0.06 & 0.93 
& $0.83_{\textcolor{gray}{47}}$ & $2.42_{\textcolor{gray}{47}}$ & $0.97_{\textcolor{gray}{0}}$ & $0.58_{\textcolor{gray}{23}}$ & $1.85_{\textcolor{gray}{0}}$ & 
\:\, 202 & 7 \\
$G1$c & $11$ & $3$ & $9$
& $\AmbPA{9}{6}{2}$
& 8 & 0.05 & 1.07
& $0.80_{\textcolor{gray}{44}}$ & $1.94_{\textcolor{gray}{44}}$ & $0.89_{\textcolor{gray}{0}}$
& $0.62_{\textcolor{gray}{6}}$ & $1.36_{\textcolor{gray}{0}}$ &
\:\, 169 & 9.5 \Bstrut \\
\hdashline\noalign{\vskip 0.5ex}
\multicolumn{3}{c}{Average} & 
&
$\AmbPA{7}{5}{1}$
& 6.3 & 0.06 & 0.92
& $0.82_{\textcolor{gray}{45}}$ & $2.63_{\textcolor{gray}{46}}$ & $1.08_{\textcolor{gray}{0}}$
& $0.62_{\textcolor{gray}{15}}$ & $2.08_{\textcolor{gray}{0}}$ &
\:\, 178 & 7.5 \\
\midrule
$G2$a & 10 & 3 & 10 
& $\AmbPA{4}{2}{2}$
& 4 & 0.07 & 0.72 
& $0.94_{\textcolor{gray}{100}}$ & $3.15_{\textcolor{gray}{100}}$ & $1.76_{\textcolor{gray}{100}}$ & $0.86_{\textcolor{gray}{100}}$ & $2.97_{\textcolor{gray}{100}}$ & 
\:\, 294 & 12 \\
$G2$b & 10 & 3 & 10 
& $\AmbPA{6}{3}{2}$ 
& 5 & 0.06 & 0.76 
& $0.68_{\textcolor{gray}{100}}$ & $1.56_{\textcolor{gray}{100}}$ & $0.83_{\textcolor{gray}{100}}$ & $0.53_{\textcolor{gray}{100}}$ & $1.26_{\textcolor{gray}{100}}$ & 
\:\, 123 & 8 \\
$G2$c & 10 & 2 & 9 
& $\AmbPA{12}{6}{3}$ 
& 9 & 0.05 & 1.20 
& $0.83_{\textcolor{gray}{100}}$ & $1.52_{\textcolor{gray}{100}}$ & $0.92_{\textcolor{gray}{100}}$ & $0.62_{\textcolor{gray}{100}}$ & $1.05_{\textcolor{gray}{100}}$ & 
\:\, 178 & 12 \Bstrut \\
\hdashline\noalign{\vskip 0.5ex}
\multicolumn{3}{c}{Average} & 
&
$\AmbPA{7}{4}{2}$
& 6.0 & 0.06 & 0.89 
& $0.82_{\textcolor{gray}{100}}$ & $2.08_{\textcolor{gray}{100}}$ & $1.17_{\textcolor{gray}{100}}$ & $0.67_{\textcolor{gray}{100}}$ & $1.76_{\textcolor{gray}{100}}$ & 
\:\, 198 & 10.7 \\
\midrule
$G3$a & 17 & 8 & 20 
& $\AmbPA{2}{1}{1}$
& 3 & 0.13 & 1.01 
& $3.39_{\textcolor{gray}{50}}$ & $10.71_{\textcolor{gray}{50}}$ & $4.16_{\textcolor{gray}{0}}$ & $2.97_{\textcolor{gray}{11}}$ & $8.75_{\textcolor{gray}{0}}$ & 
\:\, 244 & 14 \\
$G3$b & 18 & 9 & 21 
& $\AmbPA{3}{2}{1}$
& 6 & 0.10 & 1.01 
& $1.63_{\textcolor{gray}{50}}$ & $8.23_{\textcolor{gray}{50}}$ & $2.54_{\textcolor{gray}{0}}$ & $1.49_{\textcolor{gray}{11}}$ & $8.09_{\textcolor{gray}{0}}$ & 
\:\, 209 & 12.5 \Bstrut \\
\hdashline\noalign{\vskip 0.5ex}
\multicolumn{3}{c}{Average} & 
& 
$\AmbPA{2}{2}{1}$
& 4.5 & 0.12 & 1.01 
& $2.51_{\textcolor{gray}{50}}$ & $9.47_{\textcolor{gray}{50}}$ & $3.35_{\textcolor{gray}{0}}$ & $2.23_{\textcolor{gray}{11}}$ & $8.42_{\textcolor{gray}{0}}$ & 
\:\, 226 & 13.2 \\
\midrule 
$G4$a & 25 & 10 & 26 
& $\AmbPA{3}{1}{2}$
& 3 & 0.17 & 1.10 
& $2.64_{\textcolor{gray}{100}}$ & $15.20_{\textcolor{gray}{100}}$ & $3.51_{\textcolor{gray}{100}}$ & $2.63_{\textcolor{gray}{100}}$ & $8.23_{\textcolor{gray}{100}}$ & 
\:\, 123 & 3 \\
$G4$b & 25 & 8 & 24 
& $\AmbPA{12}{6}{3}$ 
& 9 & 0.12 & 1.36 
& $1.81_{\textcolor{gray}{100}}$ & $7.47_{\textcolor{gray}{100}}$ & $2.52_{\textcolor{gray}{100}}$ & $1.36_{\textcolor{gray}{100}}$ & $7.05_{\textcolor{gray}{100}}$ & 
\:\, 176 & 10.5 \\
$G4$c & 25 & 8 & 24 
& $\AmbPA{16}{6}{4}$
& 10 & 0.11 & 1.45 
& $1.67_{\textcolor{gray}{100}}$ & $7.58_{\textcolor{gray}{100}}$ & $2.35_{\textcolor{gray}{100}}$ & $1.17_{\textcolor{gray}{100}}$ & $7.13_{\textcolor{gray}{100}}$ & 
\:\, 154 & 10 \Bstrut \\
\hdashline\noalign{\vskip 0.5ex}
\multicolumn{3}{c}{Average} & 
& 
$\AmbPA{10}{4}{3}$
& 7.3 & 0.13 & 1.30 
& $2.04_{\textcolor{gray}{100}}$ & $10.09_{\textcolor{gray}{100}}$ & $2.80_{\textcolor{gray}{100}}$ & $1.72_{\textcolor{gray}{100}}$ & $7.47_{\textcolor{gray}{100}}$ & 
\:\, 151 & 7.8 \\
\midrule
$G5$a & 8 & 4 & 9 
& $\AmbPA{2}{1}{1}$
& 2 & 0.06 & 0.47 
& $0.64_{\textcolor{gray}{100}}$ & $2.09_{\textcolor{gray}{100}}$ & $0.90_{\textcolor{gray}{100}}$ & $0.49_{\textcolor{gray}{100}}$ & $1.93_{\textcolor{gray}{100}}$ & 
\:\, 80 & 4 \\
$G5$b & 8 & 4 & 9 
& $\AmbPA{2}{1}{1}$
& 3 & 0.06 & 0.61 
& $0.73_{\textcolor{gray}{100}}$ & $3.13_{\textcolor{gray}{100}}$ & $0.82_{\textcolor{gray}{100}}$ & $0.58_{\textcolor{gray}{100}}$ & $2.34_{\textcolor{gray}{100}}$ & 
\:\, 56 & 5 \\
$G5$c & 8 & 2 & 7 
& $\AmbPA{6}{3}{2}$
& 5 & 0.04 & 0.71 
& $0.49_{\textcolor{gray}{100}}$ & $0.84_{\textcolor{gray}{100}}$ & $0.58_{\textcolor{gray}{100}}$ & $0.41_{\textcolor{gray}{100}}$ & $0.71_{\textcolor{gray}{100}}$ & 
\:\, 86 & 8 \Bstrut \\
\hdashline\noalign{\vskip 0.5ex}
\multicolumn{3}{c}{Average} & 
& 
$\AmbPA{3}{2}{1}$
& 3.3 & 0.05 & 0.60 
& $0.62_{\textcolor{gray}{100}}$ & $2.02_{\textcolor{gray}{100}}$ & $0.76_{\textcolor{gray}{100}}$ & $0.50_{\textcolor{gray}{100}}$ & $1.66_{\textcolor{gray}{100}}$ & 
\:\, 74 & 5.7 \\
\midrule
$G6$a & 21 & 5 & 23 
& $\AmbPA{2}{1}{1}$
& 2 & 0.19 & 1.03 
& $2.24_{\textcolor{gray}{100}}$ & $4.80_{\textcolor{gray}{100}}$ & $2.88_{\textcolor{gray}{100}}$ & $2.33_{\textcolor{gray}{100}}$ & $5.02_{\textcolor{gray}{100}}$ & 
\:\, 84 & 5 \\
$G6$b & 21 & 5 & 23 
& $\AmbPA{14}{7}{4}$
& 11 & 0.16 & 1.75 
& $1.96_{\textcolor{gray}{100}}$ & $8.22_{\textcolor{gray}{100}}$ & $2.72_{\textcolor{gray}{100}}$ & $1.38_{\textcolor{gray}{100}}$ & $6.59_{\textcolor{gray}{100}}$ & 
\:\, 220 & 14.5 \\
$G6$c & 21 & 5 & 23 
& $\AmbPA{18}{8}{6}$
& 14 & 0.13 & 2.03 
& $2.28_{\textcolor{gray}{100}}$ & $8.32_{\textcolor{gray}{100}}$ & $3.21_{\textcolor{gray}{100}}$ & $1.54_{\textcolor{gray}{100}}$ & $6.71_{\textcolor{gray}{100}}$ & 
\:\, 266 & 17 \Bstrut \\
\hdashline\noalign{\vskip 0.5ex}
\multicolumn{3}{c}{Average} & 
& 
$\AmbPA{11}{5}{4}$
& 9.0 & 0.16 & 1.60 
& $2.16_{\textcolor{gray}{100}}$ & $7.12_{\textcolor{gray}{100}}$ & $2.94_{\textcolor{gray}{100}}$ & $1.75_{\textcolor{gray}{100}}$ & $6.11_{\textcolor{gray}{100}}$ & 
\:\, 190 & 12.2 \\
\midrule
$G7$ & 56 & 12 & 77 
& $\AmbPA{3}{2}{1}$
& 3 & 0.51 & 2.47 
& $4.79_{\textcolor{gray}{100}}$ & $24.67_{\textcolor{gray}{100}}$ & $5.60_{\textcolor{gray}{100}}$ & $4.01_{\textcolor{gray}{100}}$ & $19.45_{\textcolor{gray}{100}}$ & 
\:\, 337 & 8.5 \\
\midrule
$G8$ & 37 & 32 & 72 
& $\AmbPA{9}{3}{3}$
& 6 & 116.19 & 127.02 
& $1.92_{\textcolor{gray}{100}}$ & $5828.16_{\textcolor{gray}{100}}$ & $2.89_{\textcolor{gray}{100}}$ & $1.51_{\textcolor{gray}{100}}$ & $242.55_{\textcolor{gray}{100}}$ & 
\:\, 305 & 6.5 \\
\midrule
$G9$ & 29 & 18 & 42 
& $\AmbPA{23}{7}{9}$
& 16 & 0.24 & 2.81 
& $5.14_{\textcolor{gray}{100}}$ & $155.44_{\textcolor{gray}{100}}$ & $6.45_{\textcolor{gray}{100}}$ & $3.26_{\textcolor{gray}{100}}$ & $23.77_{\textcolor{gray}{100}}$ & 
{\tikz[remember picture]\node(vline-bot){};549} & 25 \\
\bottomrule
\end{tabular}
\end{adjustbox}
\caption{Aggregate results from \tool runs. The table presents the
    number of terminals ($|\Sigma|$), the number of nonterminals
    ($|V|$), the number of productions ($|P|$), the total number of
    ambiguities in each grammar reported by \menhir ($A_{p,a}$) where
    $p$ refers to the number of ambiguities \wrt precedence order and
    $a$ refers to the number of ambiguities \wrt associativity, the
    average number of prompts ($\Delta$), time spent in converting the
    grammar to TA in \si{\ms} (Conv), time spent for TA learning in
    \si{\ms} (Learn), TA intersection time in \si{\ms} with percentage
    of scenarios successfully disambiguated as subscript
    ($I^{\text{def}}_{\textcolor{gray}{\%}}$), TA intersection time in
    \si{\ms} without reachability-based optimisation
    ($I^{1}_{\textcolor{gray}{\%}}$), without duplicate removal
    optimisation ($I^{2}_{\textcolor{gray}{\%}}$), without epsilon
    introduction optimisation ($I^{3}_{\textcolor{gray}{\%}}$), and
    without all three optimisations
    ($I^{123}_{\textcolor{gray}{\%}}$), each with their respective
    percentage fixed as subscript, total time for manual
    disambiguation in \si{\s} ($T_{\text{man}}$), and number of
    production edits for manual disambiguation
    ($\Delta_{\text{man}}$).}
\label{tab:grammar-stats}
\end{table}
\begin{tikzpicture}[remember picture, overlay]
  \draw[gray!70, line width=0.6pt]
    ($(vline-top) + (-0.6em, 1.8ex)$)
      --
    ($(vline-bot) + (-0.6em, -0.8ex)$);
\end{tikzpicture}

\subsection{Experimental Results}
\label{subsec:results}
Aggregate results from completed runs of \tool are presented
in~\autoref{tab:grammar-stats}.
%
%
Grammars are arranged in an increasing order of size per benchmark
source, where the size can be approximated by the numbers of
productions $|P|$ as well as terminals $|\Sigma|$, and the
nonterminals $|V|$.
Variations of each grammar, should there be any, are sorted by the
number of ambiguities in an increasing order, as shown by the
$A_{p,a}$ column where $p$ is the number related to precedence order
and $a$ related to associativity. 
%
In addition to the default configuration
($I^{\text{def}}_{\textcolor{gray}\%}$), we report results under
several ablations of the intersection algorithm, shown in the
$I^{1}_{\textcolor{gray}\%}$, $I^{2}_{\textcolor{gray}\%}$,
$I^{3}_{\textcolor{gray}\%}$, and $I^{123}_{\textcolor{gray}\%}$
columns in~\autoref{tab:grammar-stats}.\footnote{The subscripted percentage {\textcolor{gray}\%} refers to success rate under each approach.
}
These ablations selectively disable optimisations introduced
in~\autoref{algo:intersect}: (i) $I^{1}$ disables
computing raw cross-products of all transitions without restricting to
reachable states via learned $\Delta$; (ii) $I^{2}$ disables
duplicate-state removal; (iii) $I^{3}$ disables
$\epsilon$-introduction; and (iv) $I^{123}$ disables all three
optimisations.
This breakdown allows us to isolate how each optimisation affects both effectiveness and efficiency of \tool. 
We address the following research questions to evaluate the results: 
\begin{itemize}
    %
\item {\tname{RQ1:}} How \emph{effective} is \tool in eliminating
  \emph{all} the ambiguities in the grammars, and what factors
  contribute to its performance?
    %
    \item {\tname{RQ2:}} How does increase in ambiguities or
    grammar size affect \emph{efficiency} of \tool?
    %
    \item {\tname{RQ3:}} How \emph{scalable} is \tool \wrt the
    input grammar size or number of ambiguities? 
\end{itemize}\vspace{-0.5\baselineskip}

\subsubsection{RQ1: Effectiveness}
\label{subsec:effectiveness}
%
%
Overall, \tool successfully repairs ambiguous grammars with an average
of $85\%$ fix rate. 
%
Seven out of ten grammars show a perfect $100\%$ fix rate, with the
lowest fix rates from $G1$ for which most failures come primarily from
the fact that \menhir, an LR(1) parser, can look ahead only \emph{one}
token at a time and is unable to distinguish differences in parsing
options that would be clear if it could look ahead further.
In other words, there are situations when \tool reports remaining
ambiguities in the repaired grammar as it relies on \menhir to
identify them, even though there are no ambiguities in the underlying
CFG. In such cases, attempts by \tool to remove forbidden trees will
not change the grammar further. 

\begin{wrapfigure}[6]{r}{.41\textwidth}
\vspace*{-15pt}
\begin{lstlisting}[basicstyle=\linespread{1.0}\ttfamily\footnotesize,numbers=none]
x2 -> IF cond THEN x2 ELSE x2
x2 -> x3
x3 -> IF cond THEN x3
 :
\end{lstlisting}
\setlength\abovecaptionskip{2pt}
\caption{Productions \wrt \code{if6} $<$ \code{if4}}
\label{fig:higherif4}
\end{wrapfigure}
These cases are reported as failures, contributing to lower fixed
rates of $G0$, $G1$ and $G3$ in~\autoref{tab:grammar-stats}.  
%
%
For example, in the case of $G0$, if a tree specifying \code{if6} $<$
\code{if4} is selected, \tool produces a grammar containing the
fragment shown in~\autoref{fig:higherif4}.
%
%
Given these productions, suppose \menhir tries to parse an expression
like \code{*if* TRUE THEN *if* TRUE THEN 1 ELSE 2}. It can be done by
first applying the rule \code{x2 -> IF cond THEN x2 ELSE x2}, and then
at the nonterminal \code{x2} after \code{THEN}, applying either (1)
the same rule \emph{or} (2) $\epsilon$-transitioning to \code{x3}
and then applying \code{x3 -> IF cond THEN x3}, even though (1) would
produce different number of \code{ELSE}s in the program. 
\begin{wrapfigure}[6]{r}{.41\textwidth}
\vspace*{-15pt}
\begin{lstlisting}[basicstyle=\linespread{1.0}\ttfamily\footnotesize,numbers=none]
x2 -> IF cond THEN x2 
x2 -> x3
x3 -> IF cond THEN x3 ELSE x3
 :
\end{lstlisting}
\setlength\abovecaptionskip{2pt}
\caption{Productions \wrt \code{if4} $<$ \code{if6}}
\label{fig:higherif6}
\end{wrapfigure}
On the other hand, if \code{if4} $<$ \code{if6} is
selected, generating transitions in~\autoref{fig:higherif6}, \menhir
no longer reports ambiguities as parsing can be done only in one way:
first, by applying \code{x2 -> IF cond x2}, then at the nonterminal
\code{x2} after \code{THEN}, $\epsilon$-transitioning to \code{x3},
and applying \code{x3 -> IF cond THEN x3 ELSE x3}.

Across all benchmarks, the success rate ({\textcolor{gray}{\%}}) of
$I^{1}_{\textcolor{gray}{\%}}$ matches that of the default
configuration ($I^{\text{def}}_{\textcolor{gray}{\%}}$). 
Although $I^{1}_{\textcolor{gray}{\%}}$ performs intersection using
unrestricted cross-products, the subsequent duplicate-state
elimination and $\epsilon$-introduction steps ensure that the
resulting automata and grammars are identical to those produced under
all optimisations. 
This indicates that optimisations~$1$ primarily improves efficiency
rather than correctness. 
In contrast, disabling duplicate-state removal
($I^{2}_{\textcolor{gray}{\%}}$) or $\epsilon$-introduction
($I^{3}_{\textcolor{gray}{\%}}$) significantly degrades effectiveness
for several grammars, notably $G0$, $G1$s, and $G3$s. 
In these cases, the resulting tree automata and CFGs remain
semantically correct, but contain redundant or overlapping states and
nonterminals that render the grammar unsuitable as input to a
deterministic LR parser generator. 
This highlights an important distinction: while tree automata and CFGs
tolerate redundancy, deterministic parsing does not. 
Consequently, optimisations such as duplicate elimination and
$\epsilon$-introduction are not merely performance improvements,
but are necessary to ensure compatibility with LR parsing, which \tool
relies on via \menhir. 
%
In addition, when all optimisations are disabled
($I^{123}_{\textcolor{gray}{\%}}$), these issues compound, leading to
consistently low success rates for the same grammar. 

For almost all benchmark data, there is at least one scenario, \ie, a
combination of user selections, which leads to elimination of all the
ambiguities.
This means, \tool can be used to offer parsing options that eventually
make the grammar conflict-free, which is potentially useful for the
users who are looking for a possible way to disambiguate the grammar.
%
It takes about $6.3$ prompts to fully fix a grammar with an average of $8$ ambiguities.
%
%
This shows the benefit of our interaction design, as it indeed
requires less than $N$ prompts to address $N$ ambiguities on average.

\subsubsection{RQ2: Efficiency} 
\label{subsec:efficiency}
%
\tool is fast in performing disambiguation across different grammars,
with an average of less than 20\si{\ms} running time (excluding the
user interaction time), which is barely noticeable to a user.
%
The time measured refers to an average amount of time it took for
successful cases, \ie, when the scenarios being tested result in
resolution of all the ambiguities. 
%
%
As shown in~\autoref{tab:grammar-stats}, the total runtime is split
into three different components, showing the time spent for conversion
(Conv), TA learning (Learn), and TA intersection with all
optimisations enabled ($I^{\text{def}}_{\textcolor{gray}{\%}}$).
%
Across most grammars, conversion contributes the least to the overall
runtime, followed by learning, with intersection typically dominating. 
%
%
%
The optimisations made in~\autoref{algo:intersect}, bring an
algorithm that is usually strictly quadratic in the number of
productions, to one that is more efficient. 
%
%
Intersection is slowest for $G7$ which has the most productions, and
for which there are many reachable states to compute.
%
Time spent for conversion and intersection appears largely unaffected
by the number of ambiguities, as shown in
in~\autoref{fig:convert-time} and \autoref{fig:intersect-time}.
%
\autoref{fig:learn-time} shows more prominent slowdown for
learning \wrt number of ambiguities, with a clearer
upward trend. 

%
$G8$, one of the larger grammars, has a significantly slower learning
and conversion time compared to the other grammars. We believe that
this is likely due to its larger number of states in combination with
a large number of productions, for which book-keeping procedures in
conversion and learning in our implementation, scale quadratically.
In contrast, the default intersection time for $G8$ remains relatively
modest. This is because intersection is computed only over reachable
state pairs, which dramatically reduces the effective search space. 
The impact of this optimisation is confirmed by the ablation results:
when reachability-based pruning is disabled ($I^{1}$ and $I^{123}$),
intersection time for $G8$ increases by several orders of magnitude. 
This shows that the reachability optimisation is important for
containing the cost of intersection in grammars with large numbers of
states and productions. 
That is, while large grammars can result in expensive conversion and
learning, the intersection optimisations are effective at preventing
state explosion from dominating total runtime. 

\subsubsection{RQ3: Scalability} 
\label{subsec:scalability}
We examine performance trends as the input grammar size (\ie, the
numbers of productions $|P|$, of terminals $|\Sigma|$, and of
nonterminals $|V|$) and ambiguity count increases to assess
scalability. 
%
Across grammar families $G1$-$G6$, increasing the number of
ambiguities has only a moderate impact on conversion or intersection
time, with a slower-than-linear upward trend in learning time, as can
be shown in~\autoref{fig:performance-metrics}.
%
%
%
This indicates that increasing the number of ambiguities has only a
moderate impact on learning time, indicating that \tool scales well
with respect to ambiguity count alone. 
%
On the other hand, grammars with large numbers of productions and
symbols or complex nonterminal structure (\eg, deeply nested or
mutually recursive nonterminals) incur higher costs, reflecting the
increased size of the induced automata. 
For instance, even for a grammar with more than 20 ambiguities (\eg, $G9$), \tool successfully eliminates all supported ambiguities within a few milliseconds. 
Moreover, \tool takes the longest on average to repair $G8$, despite
having fewer ambiguities than grammars like $G2$c, $G4$b-c, $G6$b-c
and $G9$, due to its large production set combined with its
productions containing a large number of mutually recursive
nonterminals. 
These results suggest that scalability of \tool is governed less by
the sheer number of ambiguities than by the size and the structural
complexity of the grammar.

Although these trends may suggest that \tool scales robustly in
ambiguity count and reasonably in grammar size, there is a caveat: the
times measured were obtained from successful runs, whereas a bigger
and more complex grammar is likely to contain edge scenarios which are
outside the scope of \tool or of \menhir that \tool depends on. 
This implies that scalability of \tool is subject to its scope and
limitations which we discuss further in~\autoref{subsec:scopelimits}.
%

\begin{figure*}[t]
    \centering
    \begin{subfigure}[b]{0.32\textwidth}
        \centering
        \includegraphics[width=\textwidth]{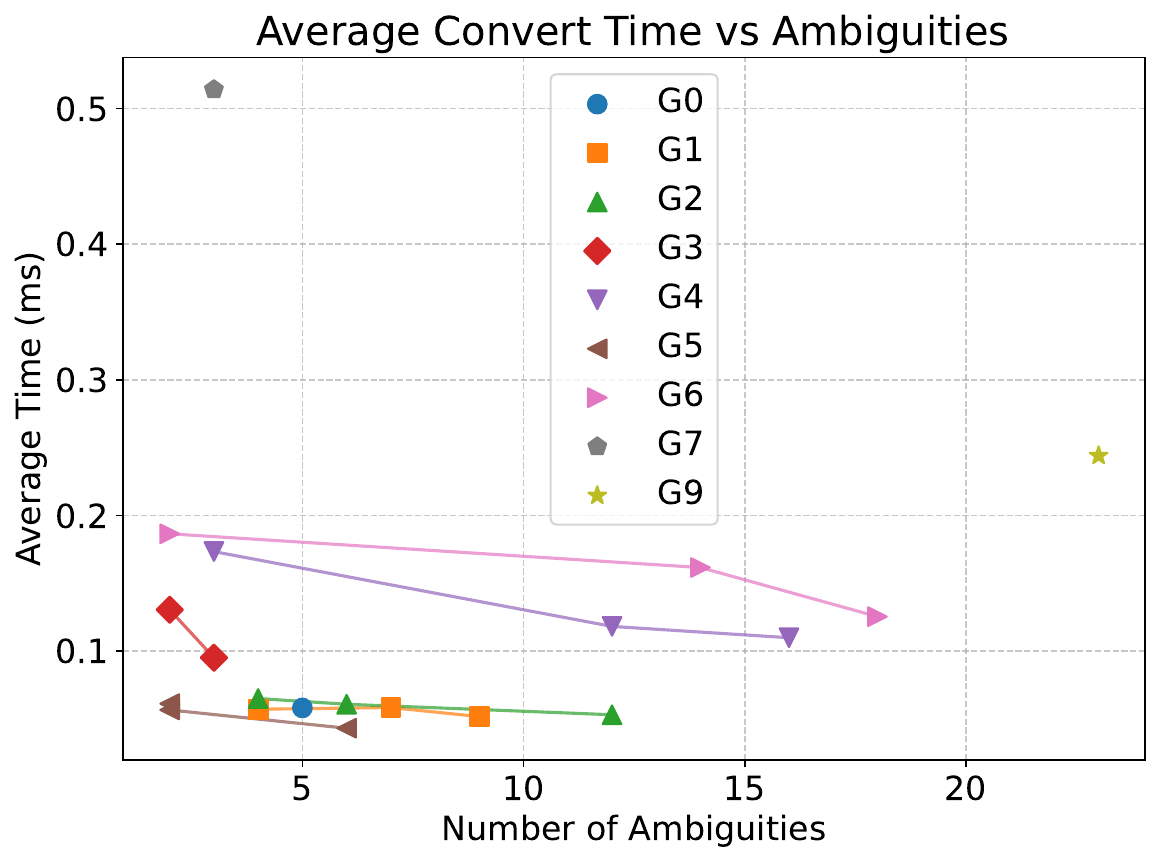}
        \caption{Average CFG to TA convert time (ms) versus number of ambiguities for different grammar types.}
        \label{fig:convert-time}
    \end{subfigure}
    \hfill
    \begin{subfigure}[b]{0.32\textwidth}
        \centering
        \includegraphics[width=\textwidth]{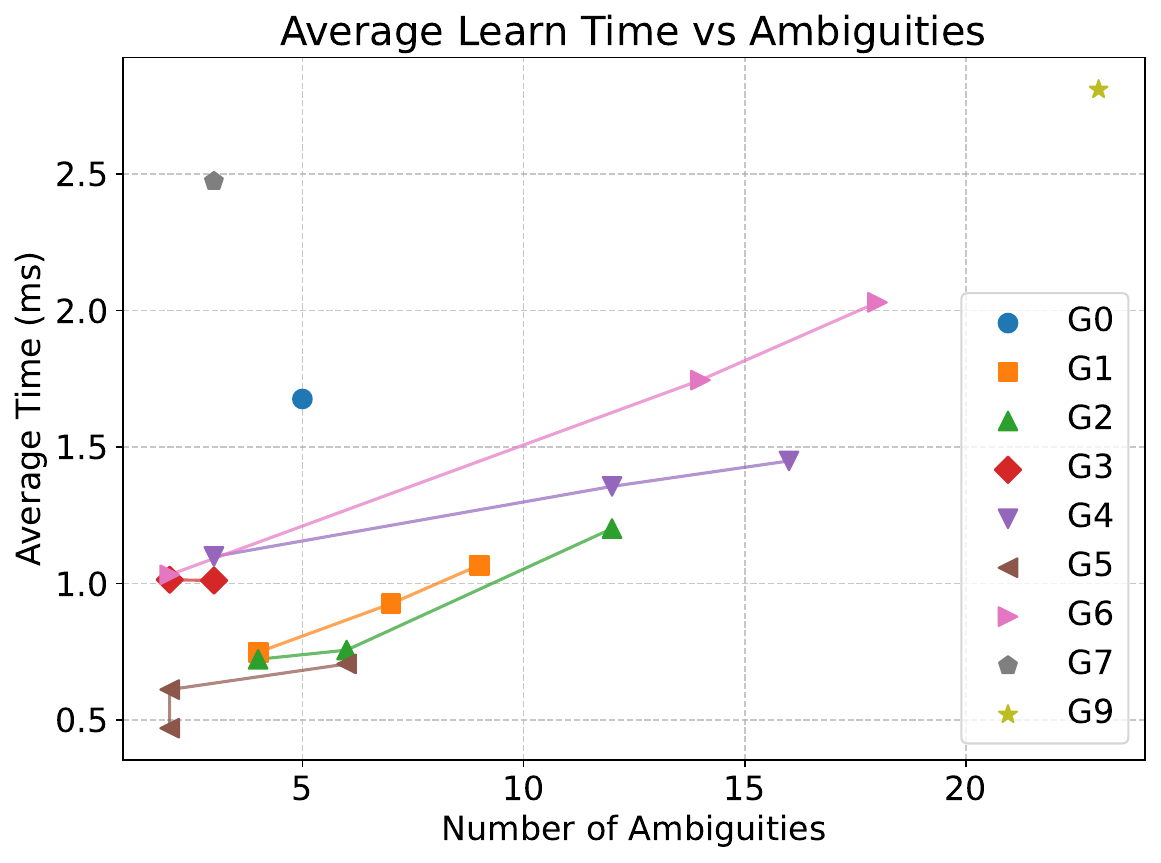}
        \caption{Average TA learning time (ms) versus number of ambiguities for different grammar types.}
        \label{fig:learn-time}
    \end{subfigure}
    \hfill
    \begin{subfigure}[b]{0.32\textwidth}
        \centering
        \includegraphics[width=\textwidth]{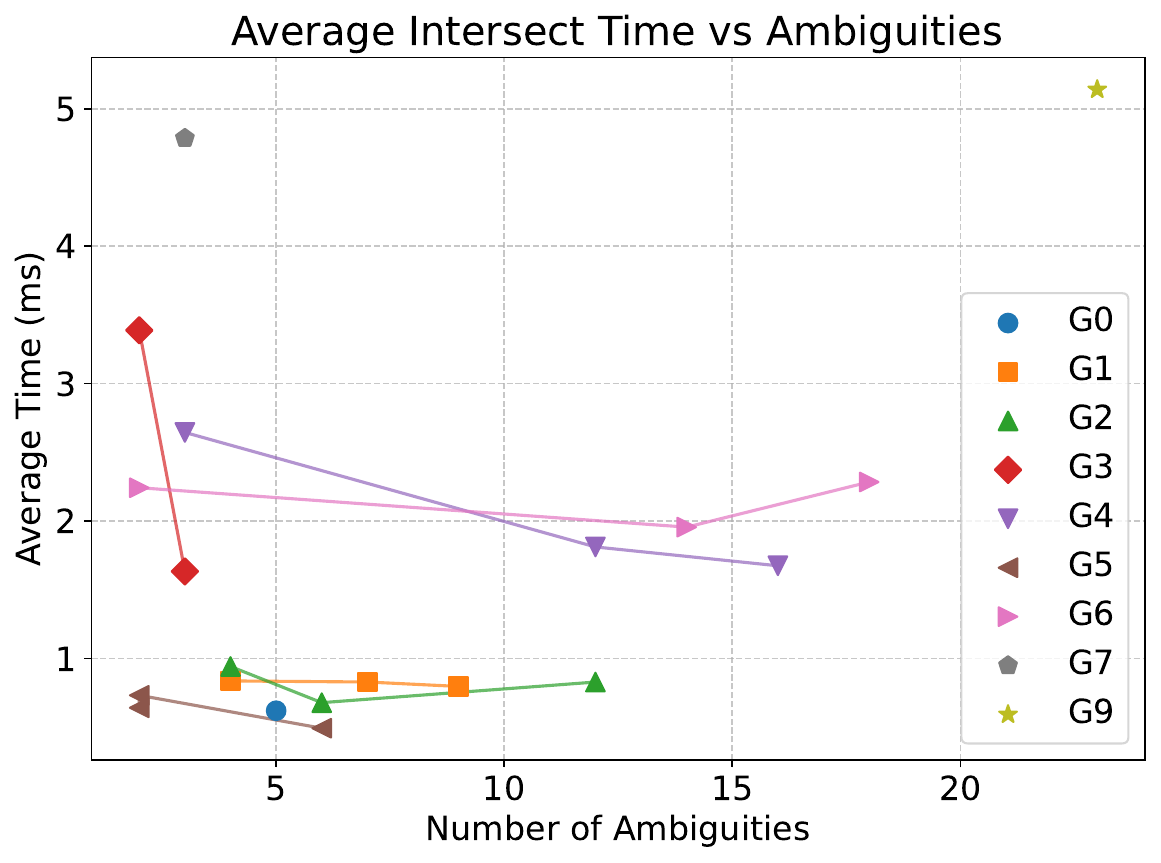}
        \caption{Average intersection time (ms) versus number of ambiguities for different grammar types.}
        \label{fig:intersect-time}
    \end{subfigure}
    
\setlength\abovecaptionskip{5pt}
\setlength\belowcaptionskip{-10pt}
    \caption{Performance metrics for various grammar types across different numbers of ambiguities. Each point represents a grammar variant, with markers indicating the grammar family.}
    \label{fig:performance-metrics}
\end{figure*}

\subsubsection{Manual Disambiguation}
\label{subsec:manual}
To assess the effort required to repair grammars \emph{without} \tool,
two of the authors independently performed manual disambiguation for
each grammar using \menhir alone. 
Manual disambiguation consisted of repeatedly inspecting \menhir's
error messages and conflict reports, editing grammar productions to
address reported ambiguities, and re-running \menhir until no further
conflicts were reported. 
Both authors were equally familiar with \menhir's input format and the
interpretation of its conflict diagnostics. 

For each grammar, we recorded the total time spent on manual repair
and the number of production edits made; the values in columns
$T_{\text{man}}$ and $\Delta_{\text{man}}$
of~\autoref{tab:grammar-stats} are the averages of the time and of
the edits made.
%
In practice, manual disambiguation takes on the order of minutes
rather than (milli)seconds and requires dozens of production edits,
even for grammars with a modest number of ambiguities.
%
Moreover, manual fixes were occasionally incorrect or overly
restrictive, unintentionally eliminating valid parses. 
Such cases required undoing or revising earlier changes, further
increasing the overall effort. 
A key contributor to this cost is the form in which ambiguities are
presented to the user: \menhir reports conflicts at the level of LR
automaton states and actions, requiring users to reconstruct the
competing parse structures and infer the intended disambiguation
before modifying the grammar. 
%
An example illustrating the user interfaces of \menhir and \tool is
provided in the supplementary, highlighting the contrast in how
ambiguities are presented to users.

\subsection{Scope and Limitations}
\label{subsec:scopelimits}

\subsubsection{\menhir vs. \tool: Scope of Ambiguities Addressed}
\label{subsec:scope}
While \tool relies on \menhir's ambiguity detection mechanisms,
\menhir reports an over-approximation of ambiguities that arise during
LR(1) parser construction, in the form of shift/reduce and
reduce/reduce conflicts. 
These conflicts may be caused by various sources such as operator
precedence and associativity, dangling-else ambiguities, grammar
underspecification, or LR(1)-specific limitations unrelated to genuine
syntactic ambiguity.
Out of these conflicts detected by \menhir as LR(1) conflicts, \tool
operates on a subset of them: \ie, ambiguities for which parsing
preferences can be captured by selecting one of the two alternative
tree examples, as described in~\autoref{subsec:learner}. 
In other words, not all conflicts reported by \menhir are addressable
by \tool, and we discuss these non-addressable conflicts in details
in~\autoref{subsec:nonaddressable}.
Accordingly, \menhir acts as a the conflict detection backend, and
\tool focuses on repairing those conflicts that can be resolved via
example-driven grammar restrictions.

\subsubsection{Termination Behaviour}
\label{subsec:termination}
\tool operates by iteratively repairing the input grammar and
re-invoking \menhir on the resulting grammar. 
The procedure terminates in one of three cases: 
(i) \menhir reports no remaining conflicts; 
(ii) the remaining conflicts do not fall within \tool's example-based ambiguity class (\autoref{subsec:learner}); or
(iii) \menhir reports conflicts that are false positives and cannot be
eliminated by further restriction of the grammar. 
In cases (ii) and (iii), \tool detects that no further progress can be
made, halts without attempting further iterations, and reports to the
user that the remaining conflicts are not addressable. 

The user cannot become stuck in an infinite loop. 
Whenever \tool applies a repair step, the resulting grammar admits a
strict subset of the parse trees admitted by the previous grammar. 
Since the original grammar yields a finite set of parse trees and
conflicts, this monotonic reduction of admissible parse trees
guarantees termination. 
We now discuss the kinds of conflicts that belong to cases (ii) and
(iii), and therefore constitute inherent limitations of \tool.

\subsubsection{Limitations and Non-addressable Conflicts}
\label{subsec:nonaddressable}
\tool is designed to resolve ambiguities that can be expressed as
\emph{example-based binary choices} (\autoref{subsec:learner},
\autoref{subsec:scope}), allowing it to offer a lightweight and
intuitive interaction model. 
However, this design choice also imposes inherent limitations on the
classes of ambiguities that can be addressed.
As described in~\autoref{subsec:termination}, there are two scenarios
where \tool may terminate without resolving all reported conflicts:
(1) false positives reported due to limitations of \menhir, (2)
ambiguities for which tree examples cannot be constructed because the
underlying parsing preferences cannot be represented as two
alternatives. 
In both cases, the remaining conflicts fall outside the scope of \tool
and are therefore reported as non-addressable.

First, the scenario (1) happens due to \tool's dependency on \menhir,
as mentioned in~\autoref{subsec:effectiveness}. 
%
A natural way we can address this issue is by employing a parser that
is not limited by the number of lookahead tokens. 
For example, \tname{dypgen}~\cite{Onzon:dypgen2012} is a GLR parser
generator not limited to a fixed number of lookahead tokens.
Alternatively, \tname{Earley}~\cite{Raffalli:earley2020} allows a
construction of a parser with an unbounded number of lookahead tokens
based on an algorithm that creates a chart of possible parses via
dynamic programming.
While these are viable options, a potential issue would be their
inefficiency (both have cubic time complexity
compared to linear time of \menhir), and, as for \tname{Earley}, the
input CFG needs to be formatted with a combination of functions,
requiring more work to define a parser.
Moreover, scenario (2) can happen when examples cannot be constructed
to form trees as per~\autoref{subsec:learner}. Informally, this can be
understood as 3 or more different productions contributing to an
ambiguity, making it difficult to show parsing preferences as two
different alternatives.



\section{Related Work}
\label{sec:related}

\paragraph{Grammar repair}
%
%
Our work was inspired by the approach of Adams and
Might~\cite{Adams-Might:OOPSLA17}.
%
%
%
%
%
In their work, tree regular expressions capturing undesired sets of
parse trees are taken as input and are intersected with a tree
automaton corresponding to the input ambiguous CFG.
This results in a procedure whose time complexity grows exponentially
in the number of negative examples \emph{and} grammar size (due to negation). 
In our approach, in addition to
providing a simpler interface for the user than having to write formal
tree regular expressions, the learning algorithm from
\autoref{subsec:learner} constructs a \emph{single} tree automaton
encoding all user preferences, and performs a single intersection step
which is quadratic in the size of the automata regardless of the
number of examples. Additionally, our automaton learning step is
linear in the number of examples, resulting in a significantly more
efficient and scalable approach.

Other existing works present ways to resolve grammar
ambiguities~\cite{Klint:ASMICS94, deSouza:TASEP18, Visser:IWPT97} or
repair grammars~\cite{Raselimo-Fischer:SLE21}.
Grammars can be repaired with declarative disambiguation rules called
\emph{filters}~\cite{Klint:ASMICS94, deSouza:TASEP18, Visser:IWPT97}
that remove undesired parse trees either as a post-parse step, or
embedded in the generated parser. Unlike our approach, filters do not
provide a way to update a grammar incorporating the disambiguation
rules. This limits the interpretability of the interactions between
the parsing restrictions and the original CFG, since one needs to
regularly reference disambiguating rules to understand what parse
trees are really accepted by a CFG.
%
%
Another automated grammar repair
technique~\cite{Raselimo-Fischer:SLE21} addresses the problem of
repairing the grammar that fails tests from a provided test suite,
assuming that a correct grammar exists. It fixes the erroneous grammar
by iteratively applying a set of bespoke patches. \tool does not rely
on a fixed set of rewrite rules, and instead utilises a general
mechanism of tree automata derived from user-chosen examples.

%


\paragraph{Tree automata learning via Angluin's algorithm}
A well-known approach to learn deterministic finite automata (DFAs) is
Angluin's L* algorithm~\cite{Angluin:IC81,Angluin:IC87}, which relies on
an oracle, and maintains an observation table through oracle queries, 
which is later used to construct the DFA.
%
While there is existing work that applies the idea of the L* algorithm
for learning a TA~\cite{Besombes-Marion:TCS07}, we found it difficult
to adopt this approach in our setting.
%
%
This is because the L* algorithm assumes the existence of an oracle
that can answer membership queries (\ie, if a tree belongs to a TA),
%
%
and providing such an oracle would be
challenging in our programming by example setting.


We initially considered treating interactions with the user as an
oracle, and involve user prompts in the learning loop.
This approach, however, turned out to be neither desirable nor
necessary because the L* algorithm for tree automata learning requires
a representative set of example trees, so that every TA transition
exercises at least one example. This would require (a)~constructing
examples involving symbols that are not involved in any conflicts and
(b)~numerous interactions with the user in order to correctly simulate
an oracle, defeating the purpose of our tool, especially when the size
of the given grammar is large.
%
%
Our learning algorithm, on the other hand, does not need to involve
the user other than when they are asked to specify their preferred
tree examples \wrt ambiguities.
Alternatively, we could consider the input CFG---or, the TA translated
from the CFG, to be precise---as an oracle.
Such an oracle, however, could not correctly answer membership queries
about the additional user-provided restrictions.
%
%


\paragraph{Programming by example}
%
Programming by example (PBE) entails the synthesis of programs from a
set of input-output examples. PBE has been adopted in synthesising
grammars in Leung \etal's work~\cite{LeungSL15}, which constructs
parsers from user-provided parse tree examples.
This work encodes parsing restrictions through specification of
associativity and precedence order rules, similar to the construction of 
grammar restriction rules in \tool. These rules are then used to remove 
unwanted parse trees.
This approach requires the user to provide examples that fully
characterise the intended language, since the methodology does not
permit providing an existing grammar as context. For large languages,
or for expanding existing grammars, this can get quickly impractical.
Our technique uses an existing grammar as context, with user
preferences as additional constraints to produce an updated grammar.
This simplifies the process of writing a grammar from scratch, but
also allows for incremental updates and maintenance of existing
grammars.

Wang \etal~\cite{Wang-Dillig-Singh:OOPSLA17} implement PBE for learning TAs
for data completion tasks in tabular data. Similar to other PBE approaches, 
this requires user-provided examples, constrained by formulae in a
domain-specific language (DSL) for reasoning about tabular data. 
While in Wang \etal's work tree automata are employed for compact
representations of the user examples, the problem tackled in that work
is fundamentally different from ours in that the approach is
specialised for tabular data completion tasks, and relies heavily on
the user providing correct and sufficient examples.



\section{Conclusion}
\label{sec:conclusion}

In this work, we presented a novel take on the problem almost as old
as the area of programming languages itself---specifying syntax of a
programming language in a way that is deterministic and is free from
parsing ambiguities~\cite{AhoJU73}.
Specifically, we have cast the problem of disambiguating a
context-free grammar as an instance of \emph{programming by example},
structuring the process of grammar repair (\ie, removing ambiguities)
as a series of lightweight interactions with the grammar designer, who
guides the repair by choosing their preferred parse trees.

We believe that the key idea of our approach---compiling pairs of
complementary positive/negative tree examples into tree automata used
to refine an initially provided ``base'' grammar---has applications
beyond just resolving ambiguities in context-free grammars.
In particular, tree automata-based representation of examples can be
used as a tool to guide the design of formal grammars, benefitting
both programming language designers and automated tools for grammar
learning.

\begin{acks}
  We thank Hila Peleg and George P\^{i}rlea for their feedback on drafts of this
  paper.
  We also thank the anonymous OOPSLA'26 reviewers.
  This work was partially supported by a Singapore Ministry of
  Education (MoE) Tier~3 grant ``Automated Program Repair''
  MOE-MOET32021-0001.
\end{acks}

\ifext{
  \section*{Data Availability}
\label{sec:data-availability}

The software artefact accompanying this paper is available
online~\cite{greta-artifact}.
It contains the OCaml implementation of \greta, including
\textsc{GenTA} (\autoref{subsec:learner}) and \textsc{IntersectTA}
(\autoref{subsec:intersect}), as well as benchmark data and build
scripts for reproducing the evaluation results reported
in~\autoref{sec:eval}. 

  \bibliography{references}

  \clearpage
  \appendix
  \section{Formalism}
\label{sup:formalism}

This section provides formal definitions of the core concepts used in
\tool: context-free grammars, tree automata, tree acceptance, and the
translation from CFGs to tree automata.

\subsection{Context-Free Grammars}
\label{sup:cfg}

\begin{definition}[Context-Free Grammar]
A \emph{context-free grammar} (CFG) is a tuple $\mathcal{G} = (V, \Sigma, S, P)$ where:
\begin{itemize}
  \item $V$ is a finite set of \emph{nonterminal symbols},
  \item $\Sigma$ is a finite set of \emph{terminal symbols}, with $V \cap \Sigma = \emptyset$,
  \item $S \in V$ is the \emph{start symbol}, and
  \item $P \subseteq V \times (V \cup \Sigma)^*$ is a finite set of \emph{productions}.
\end{itemize}
\end{definition}

We write a production $(A, \beta) \in P$ as $A \to \beta$, where
$A \in V$ is the \emph{left-hand side} and $\beta \in (V \cup \Sigma)^*$
is the \emph{right-hand side} of the production.

\begin{definition}[Parse Tree]
Given a CFG $\mathcal{G} = (V, \Sigma, S, P)$, a \emph{parse tree} is a
finite ordered tree $t$ where:
\begin{itemize}
  \item Each internal node is labeled with a nonterminal $A \in V$,
  \item Each leaf node is labeled with a terminal $a \in \Sigma$,
  \item For each internal node labeled $A$ with children labeled
        $X_1, \ldots, X_n$ (from left to right), there exists a
        production $A \to X_1 \cdots X_n$ in $P$.
\end{itemize}
A parse tree is \emph{complete} if its root is labeled with the start
symbol $S$. We write $\mathcal{L}_{\mathcal{G}}$ for the set of all
complete parse trees of $\mathcal{G}$.
\end{definition}

\begin{definition}[Ambiguous Grammar]
A CFG $\mathcal{G}$ is \emph{ambiguous} if there exists a string
$w \in \Sigma^*$ such that $w$ is the yield of two or more distinct
parse trees in $\mathcal{L}_{\mathcal{G}}$.
\end{definition}

\vspace{5pt}
\subsection{Tree Automata}
\label{sup:ta}

We work with a variant of finite tree automata adapted for representing
CFG parse trees. In this formulation, transitions consume both states
(corresponding to nonterminals) and terminal symbols, preserving the
structure of CFG productions.

\begin{definition}[Ranked Alphabet]
A \emph{ranked alphabet} is a finite set $\mathcal{F}$ of symbols,
each associated with a non-negative integer called its \emph{rank} (or
arity). That is, $\mathcal{F} \eqdef \{(\textit{Sym}(p),
\textit{Rank}(p))~|~p \in P \}$. 
We write $\mathcal{F}^{(n)}$ for the set of symbols in $\mathcal{F}$
with rank $n$.
%
\end{definition}

For a CFG production $p: A \to X_1 \cdots X_n$, we define:
\begin{itemize}
  \item $\textit{Sym}(p)$ is a unique symbol associated with the production.
  For sake of presentation in this paper, we use the first terminal symbol 
  in $X_1 \cdots X_n$, or a distinguished symbol $\delta$ if there are no 
  terminals,
  \item $\textit{Rank}(p) = n$, the length of the right-hand side.
\end{itemize}

\begin{definition}[Tree Automaton]
\label{def:ta}
A \emph{(finite) bottom-up tree automaton} is a tuple
$\mathcal{A} = (Q, \mathcal{F}, \Sigma, Q_{\textup{f}}, \Delta)$ where:
\begin{itemize}
  \item $Q$ is a finite set of \emph{states},
  \item $\mathcal{F}$ is a ranked alphabet of \emph{constructor labels};
        we assume $(\epsilon, 1) \in \mathcal{F}$ to allow for
        $\epsilon$-transitions,
  \item $\Sigma$ is a finite set of \emph{terminal symbols}, with $Q \cap \Sigma = \emptyset$,
  \item $Q_{\textup{f}} \subseteq Q$ is a set of \emph{final} (accepting) states,
  \item $\Delta$ is a finite set of \emph{transition rules}. For a
        ranked symbol $f \in \mathcal{F}^{(n)}$, transitions labeled
        by $f$ are functions in
        $(Q \cup \Sigma)^{n} \times \{f\} \rightarrow Q$.
\end{itemize}
We write a transition as:
\[
  q \leftarrow_f k_1, k_2, \ldots, k_n
\]
where $f \in \mathcal{F}^{(n)}$, $q \in Q$, and each $k_i \in Q \cup \Sigma$.
We call $q$ the \emph{target state}, $f$ the \emph{constructor label},
and $k_1, \ldots, k_n$ the \emph{children} of the transition.
\end{definition}

The transition $q \leftarrow_f k_1, \ldots, k_n$ indicates that when
processing a node labeled with constructor $f$ whose children (from
left to right) are states or terminals $k_1, \ldots, k_n$, the node
transitions to state $q$.

%
Moreover, the above TA definition corresponds to a slightly
modified version of the usual TA definition whose $\Delta$ consists of
transition rules from a combination of terminals and/or states
(corresponding to terminals and/or nonterminals on the right-hand side
of productions in a CFG) \textit{and} a constructor label to a state
(corresponding to the left-hand side of the CFG production).\footnote{
We convert a CFG to a bottom-up TA. In the case of top-down TAs, the
rules transition in opposite direction. }
As such, the structure of the productions $P$ in the CFG is preserved
in the transition rules $\Delta$ of the TA. This formulation is
lightweight due to the need for fewer transitions and states compared
to the conventional definition, subsequently useful in TA intersection.

\begin{definition}[$\epsilon$-Transition]
An \emph{$\epsilon$-transition} is a transition of the form
$q \leftarrow_{(\epsilon, 1)} q'$ where $q, q' \in Q$. This transition
allows state $q'$ to be treated as state $q$ without consuming any
tree node.
\end{definition}
 
TAs can include $\epsilon$-transitions referring to a transition that
does not consume any symbol. In this paper, we include an epsilon
symbol $(\epsilon, 1)$ in $\mathcal{F}$ when $\Delta$ includes
$\epsilon$-transitions in order to make such a transition explicit.

\vspace{5pt}
\subsection{Tree Acceptance}
\label{sup:acceptance}

We define acceptance for bottom-up tree automata, which process trees
starting from the leaves and moving toward the root.

\vspace{3pt}
\begin{definition}[Run]
\label{def:run}
Given a tree automaton $\mathcal{A} = (Q, \mathcal{F}, \Sigma, Q_{\textup{f}}, \Delta)$
and a tree $t$, a \emph{run} of $\mathcal{A}$ on $t$ is a mapping
$\rho: \textit{Nodes}(t) \to Q$ that assigns a state to each node of $t$,
satisfying the following conditions:
\begin{enumerate}
  \item For each internal node $v$ with constructor label $f \in \mathcal{F}^{(n)}$
        and children $v_1, \ldots, v_n$, there exists a transition
        $\rho(v) \leftarrow_f k_1, \ldots, k_n \in \Delta$ where for each $i$:
        \begin{itemize}
          \item if $v_i$ is an internal node, then $k_i = \rho(v_i)$ or
                there is a sequence of $\epsilon$-transitions from $\rho(v_i)$ to $k_i$,
          \item if $v_i$ is a leaf labeled with terminal $a$, then $k_i = a$.
        \end{itemize}
\end{enumerate}
\end{definition}

\vspace{3pt}
\begin{definition}[Acceptance]
\label{def:acceptance}
A tree $t$ is \emph{accepted} by a tree automaton
$\mathcal{A} = (Q, \mathcal{F}, \Sigma, Q_{\textup{f}}, \Delta)$ if there
exists a run $\rho$ of $\mathcal{A}$ on $t$ such that $\rho(r) \in Q_{\textup{f}}$,
where $r$ is the root of $t$.
The \emph{language} of $\mathcal{A}$, written $L(\mathcal{A})$, is the
set of all trees accepted by $\mathcal{A}$.
\end{definition}

\vspace{5pt}
Intuitively, to check if a tree is accepted:
\begin{enumerate}
  \item Start at the leaves; terminal symbols are matched directly.
  \item For each internal node, find a transition rule matching the
        constructor label and the states/terminals of the children.
  \item If such a rule exists, assign its target state to the node.
  \item The tree is accepted if the root can be assigned a final state.
  \item $\epsilon$-transitions allow a state to be ``promoted'' to
        another state without consuming a node.
\end{enumerate}
\begin{wrapfigure}[11]{r}{0.19\textwidth}
\vspace{-18pt}
  \centering
  \includegraphics[width=0.18\textwidth]{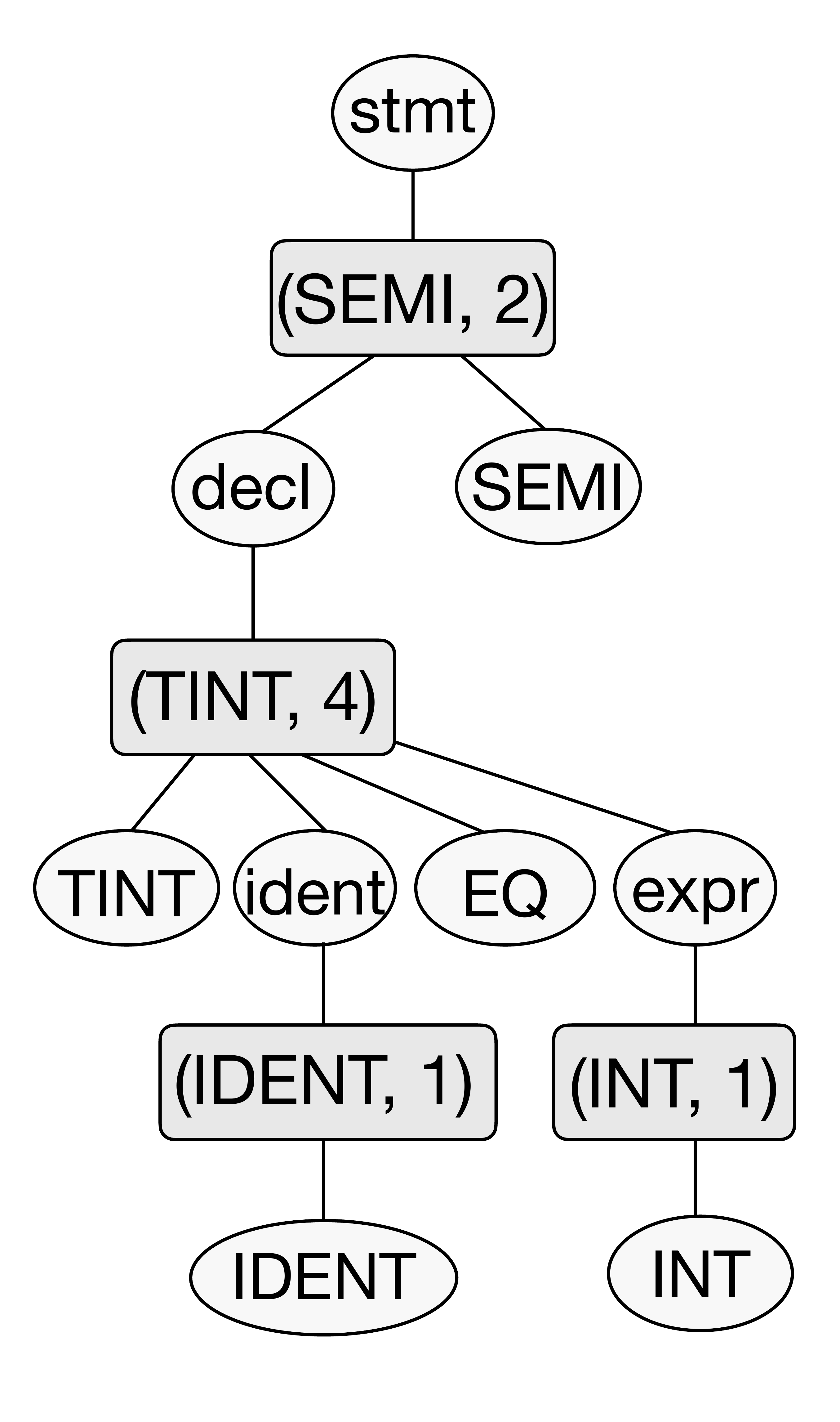}
\setlength\abovecaptionskip{5pt}
  \caption{Tree example.}
  \label{fig:trees}
\end{wrapfigure}
%

%
%
%
%
%
For example, \autoref{fig:trees} illustrates a tree example accepted
by $\mathcal{A}\textsubscript{g}$ discussed in the paper, which
represents a way to parse an expression \token{TINT IDENT EQ INT
SEMI}.
The tree is accepted because there exists a successful run of the
automaton $\mathcal{A}\textsubscript{g}$ on the tree applying the
following four rules in the bottom-up order:
(i) \code{ident <-_ident1- IDENT}, (ii) \code{expr <-_int1- INT},
(iii) \code{decl <-_tint4- TINT ident EQ expr}, and (iv) \code{stmt 
<-_semi2- decl SEMI}.
%

\vspace{5pt}
\subsection{Converting CFGs to Tree Automata}
\label{sup:cfg-to-ta}

Any CFG can be converted to a tree automaton that accepts exactly
its parse trees. We now make this construction precise.

\begin{definition}[CFG to TA Translation]
\label{def:cfg-to-ta}
Given a CFG $\mathcal{G} = (V, \Sigma, S, P)$, we construct a tree
automaton $\mathcal{A}_{\mathcal{G}} = (Q, \mathcal{F}, \Sigma, Q_{\textup{f}}, \Delta)$
as follows:
\begin{enumerate}
  \item \textbf{States:} $Q = V$ (each nonterminal becomes a state).

  \item \textbf{Final states:} $Q_{\textup{f}} = \{S\}$ (the start symbol
        is the unique accepting state).

  \item \textbf{Ranked alphabet:} For each production $p \in P$, define
        the ranked symbol $(\textit{Sym}(p), \textit{Rank}(p))$ and let
        \[
          \mathcal{F} = \{(\textit{Sym}(p), \textit{Rank}(p)) \mid p \in P\}
        \]
        We write $\textit{Prod}: \mathcal{F} \to P$ for the function
        mapping each ranked symbol back to its production.

  \item \textbf{Transitions:} For each production $p: A \to X_1 \cdots X_n$
        with ranked symbol $f = (\textit{Sym}(p), n)$, add the transition:
        \[
          A \leftarrow_f X_1, X_2, \ldots, X_n
        \]
        where each $X_i$ is either a state (if $X_i \in V$) or a
        terminal (if $X_i \in \Sigma$).
\end{enumerate}
\end{definition}

\begin{theorem}[Correctness of Translation]
\label{thm:cfg-ta-equiv}
For any CFG $\mathcal{G}$, the tree automaton $\mathcal{A}_{\mathcal{G}}$
constructed by Definition~\ref{def:cfg-to-ta} satisfies:
\[
  L(\mathcal{A}_{\mathcal{G}}) = \mathcal{L}_{\mathcal{G}}
\]
That is, $\mathcal{A}_{\mathcal{G}}$ accepts exactly the complete parse
trees of $\mathcal{G}$.
\end{theorem}

\begin{proof}
Follows directly from the construction.
\end{proof}

\vspace{6pt}
\section{Proofs of Theorems}
\label{sup:proofs}

\begin{lemma}[Relation of $O_p$ and $\mathcal{L}_r$]
\label{lemma:oplr-relation}
    For some $(s_1, o_1), (s_2, o_2) \in O_p$ where $s_1$ and $s_2$ are involved in a precedence conflict, we have that 
    $$
        o_1 \leq o_2 \iff \text{$\mathcal{L}_r$ includes trees with $s_1$ directly above $s_2$}
    $$
\end{lemma}

\begin{proof} 
$ $ \newline
    ($\Rightarrow$) In the constructed tree automaton $\mathcal{L}_r$, the transition corresponding to $s_1$ will lead to state $e_{o_1}$ and the transition corresponding to $s_2$ will lead to state $e_{o_2}$. Since $o_1 \leq o_2$, it is possible for any state $e_{o_1}$ on the right hand side of the $s_1$ transition to come from state $e_{o_2}$, through only $\epsilon$-transitions. Since the start state is also reachable from $e_{o_1}$ by epsilon transitions, such a tree with $s_1$ directly above $s_2$, will be in $\mathcal{L}_r$.\\
    ($\Leftarrow$) If $o_1 > o_2$, then since, epsilon transitions are only introduced from higher-leveled states to lower-leveled states, and right hand sides of transitions at level $i$ only mention states $\geq i$ (with the exception of introduced cycles), it is not possible for the state $s_2$ to appear directly below in the tree from $s_1$. Cases of symbols at adjacent levels which are involved in a conflict (i.e. they are involved in a cycle) are explicitly not handled by the algorithm.
\end{proof}

\begin{lemma}[Preservation of Orders of symbols not involved in Conflicts]
\label{lemma:preservation}
    For some $(s_1, o_1), (s_2, o_2) \in O_{bp}$ where either $s_1$ or $s_2$ are not involved in a precedence conflict such that $s_1$ can appear above $s_2$ in some parse tree in $\mathcal{L}_g$ and $o_1 \leq o_2$. Then there exists $(s_1, o_1'), (s_2, o_2') \in O_{p}$ such that $o_1' \leq o_2'$
\end{lemma}

\begin{proof}
    If $o_1 < o_2$, then both symbols are not involved in conflicts and as a result of $O_p$ learning, maintain their order in $O_p$. That is, there are new orders in $O_p$, $o_1'$ and $o_2'$ are such that $o_1' < o_2'$.

    If $o_1 = o_2$, and since by assumption, at least one of these symbols is not involved in a conflict, again by the construction of $O_p$, which copies the non-conflicting symbols to each newly inserted orders, there remain orders $o_1'$ and $o_2'$ such that $o_1' = o_2'$.  
\end{proof}

{\parindent0pt
\textbf{Theorem 3.1}~~(Soundness of \textsc{GenTA})
\hskip 1em}
Let $\mathcal{A}_{\text{r}}$ be the finite TA returned by GenTA and
$\mathcal{A}_{\text{g}}$ be the TA derived from CFG $\mathcal{G}$,
$T^{-}$ be the negative tree examples.
Further, $\mathcal{L}_{\text{r}} = L(\mathcal{A}_{\text{r}})$,
$\mathcal{L}_{\text{g}} = L(\mathcal{A}_{\text{g}})$, and
$\mathcal{L}^{-} = \bigcup\limits_{t \in T^{-}} P^{-}(t)$.
%
Then, the following statements hold: 
\begin{enumerate}
  {\item $\mathcal{L}_r \supseteq \mathcal{L}_{\text{g}} \setminus \mathcal{L}^-$}
  {\item $\mathcal{L}_r \cap \mathcal{L}^- = \emptyset$.}
\end{enumerate} 
\begin{proof}
  We provide proofs to the two statements. 
  \begin{enumerate}
    {\item Proof by combination of~\autoref{lemma:oplr-relation} and~\autoref{lemma:preservation}. 
    For all trees $t \in \mathcal{L}_g \setminus \mathcal{L}^-$, we have that for each of the direct-parent child symbol ($s_1$, $s_2$) relationships, either (1) both symbols are involved in $T^-$ but are correctly oriented, (2) $s_1$ appears directly above $s_2$, but $o_1 > o_2$ (there is a cycle), or (3) they satisfy the precondition for~\autoref{lemma:preservation}.

    In case (1), since the $O_p$ learning algorithm maintains the total ordering of conflicting symbols, the correct order of $s_1$ and $s_2$ is maintained in $O_p$. As a result, such relationships are preserved in $L_r$. This can be seen with a similar argument as the ($\Rightarrow$) direction of~\autoref{lemma:oplr-relation}

    For case (2), the cyclic productions introduced by GenTA maintains transitions that consume such parent-child relationships.

    Finally, in the case that at least one symbol is not involved in $T^-$, and $o_1 \leq o_2$ in $O_{bp}$, by~\autoref{lemma:preservation}, there exists $o_1', o_2'$ in $O_p$ such that $o_1' \leq o_2'$. Then by a similar argument to the ($\Rightarrow$) direction of~\autoref{lemma:oplr-relation}, there exist transitions in $\mathcal{A}_r$ that consume such nodes in the tree. 

    Since for all direct parent-child nodes in the tree $t$, our tree automaton has transitions that consume it, it must be that the tree $t$ is accepted by the automaton, implying that $t \in \mathcal{L}_r$
    
    }
    {\item Consider $T^{-} = T^{-}_{a} \cup T^{-}_{p}$. 
    First, $\forall t_{p} \in T^{-}_{p}$ where $s_{1}$ is the top symbol and $s_{2}$ is the bottom symbol, $O_{p}$ is added with $(s_{1}, o_{1})$ and $(s_{2}, o_{2})$ where $o_{1} > o_{2}$. 
    Then, based on~\autoref{lemma:oplr-relation}, $P^{-}(t_{p}) \cap \mathcal{L}_{r} = \emptyset$.
    Next, $\forall t_{a} \in T^{-}_{a}$ where $s$ is a symbol in $t$, $idx = t_{idx}$, and $(s, i) \in O_{p}$, $\delta_{s} \eqdef \delta(e_{i}, [ \ldots; e_{i+1}; \ldots ], s)$ is added to $\Delta_{r}$ where $e_{i+1}$ occurs at $idx$.
    This means \textit{ParseTrees}$(t) \cap \mathcal{L}_{r} = \emptyset$. Therefore, we have that $\mathcal{L}_{r} \cap \mathcal{L}^{-} = \emptyset$.
    }
  \end{enumerate}
\end{proof}

{\vspace{5pt}
\parindent0pt
\textbf{Theorem 3.2}~~(Correctness of \tool)
\hskip 1em}
The intersection of automaton $\mathcal{A}_{\text{r}}$
(\textsc{GenTA}'s result) with $\mathcal{A}_{\text{g}}$ (the automaton
derived from CFG $\mathcal{G}$) produces a tree automaton recognizing
the language $\mathcal{L}_{\text{g}} \setminus \mathcal{L}^{-}$.

\begin{proof}
  The resulting automaton from the intersection accepts the tree
  language $\mathcal{L}_{\text{r}} \cap \mathcal{L}_{\text{g}}$, which
  together with Theorem $3.1$ implies that the
  language it recognizes is exactly $\mathcal{L}_{\text{g}} \setminus
  \mathcal{L}^{-}$. 
\end{proof}

\section{Algorithm Complexity Analysis}
\label{sup:complexity}

\subsection{Algorithm 3.1 \textsc{LearnOaOp}}
\label{sup:algo-learnoaop}
\textbf{Algorithm 3.1} essentially re-inserts the conflicting symbols (of at most size $|\mathcal{F}|$) at different orders, while updating the orders of other symbols in $\mathcal{F}$, so in the worst case, could have $O(|\mathcal{F}|^2)$ time complexity, and $O(|\mathcal{F}|)$ space complexity.

\subsection{Algorithm 3.2 \textsc{GenTA}}
\label{sup:algo-genta}

\begin{itemize}
  \item Let $|\mathcal{F}|$, $|O_{p}|$, $|O_{a}|$ be sizes of ranked symbols, and the sets $O_{p}$ and $O_{a}$, and let $r_{\max}$ be a max rank in $\mathcal{F}$.

  \item \textbf{States $Q$}: The algorithm populates $Q$ with $m$---max order of $O_{p}$---states, so $|Q| = m$.

  \item \textbf{Build $\delta_{\mathcal{F}}$ from $P$}: The delta generator retrieves a template from a predefined productions map which is built once outside the algorithm, with space cost of $O(r_{\max} \cdot |\mathcal{F}|)$. And since the RHS of productions is part of the grammar, the only space overhead associated with $\delta_{\mathcal{F}}$ is the productions map, so $O(|\mathcal{F}|)$. Also, lookup of a symbol on the map is $O(1)$ and traversing RHS of the corresponding production is proportional to the production length, which is at most the max arity of the ranked symbols, so time complexity of $O(r_{\max})$.

  \item \textbf{Loop over $O_{p}/O_{a}$}: Inside the loop over $(s,i) \in O_{p}$, it does constant-time membership checks in $O_{a}$ and calls the $\delta$-generator. This results in time complexity of $O(r_{\max} \cdot |O_{p}|)$ and space complexity of $O(|O_{p}|)$.

  \item \textbf{Loop over $\mathcal{F}_{tr}$ and $[1..m-1]$}: This part makes a pass over $\mathcal{F}_{tr}$ and a pass over levels $1, \ldots, m-1$, resulting in time and space complexity of $O(|\mathcal{F}_{tr}| + m)$.

  \item \textbf{Loop over symbols in $\mathit{HighToLow}(G, O_{p})$}: $\mathit{HighToLow}(G,O_{p})$ is bounded by $O(|\mathcal{F}|)$. And this part calls $\delta_{\mathcal{F}}$ inside the loop, resulting in time complexity of $O(r_{\max} \cdot |\mathcal{F}|)$, and the space cost is $O(|\mathcal{F}|)$.

  \item Combining the above together, Algorithm 3.2 has the time and space complexity of $O(|\mathcal{F}|)$.
\end{itemize}

\subsection{Algorithm 3.3 \textsc{IntersectTA}}
\label{sup:algo-intersect}

\begin{itemize}
  \item Let $|\Delta_{g}|$, $|\Delta_{r}|$, and $|\mathcal{F}|$ be sizes of transitions of the input TAs, the ranked symbols, and let $r_{\max}$ be a max rank in $\mathcal{F}$.

  \item \textbf{Initialisation of $Q_{f}, Q, Q_{tmp}$}: Each of the input automata has a single accepting state, taking constant time to compute.

  \item \textbf{Loop over reachable states}:
  \begin{itemize}
    \item Each product state enters $Q_{tmp}$ \emph{at most} once, where the number of iterations is bounded by $|Q| \leq |Q_{g}| \cdot |Q_{r}|$.

    \item Then, looking up reachable symbols for each product state takes $|\mathcal{F}|$ time, and iterating over these symbols has time complexity of $O(|\mathcal{F}|)$ because all symbols can be reachable in the worst case.

    \item For a reachable symbol and a reachable product state $(e_{g}, e_{r})$, the number of candidate pairs of transitions is $|\Delta_{g}| \cdot |\Delta_{r}|$ in the worst case. Note that even though the for loop over the reachable symbols is within reachable product states, overall number of pairs of transitions is still bounded by $|\Delta_{g}| \cdot |\Delta_{r}|$, hence no need to multiply by $|\mathcal{F}|$.

    \item For each pair of transitions, the algorithm takes cross product of the transitions, which is bounded by max arity, thus at the cost of $O(r_{\max})$.

    \item Under the bounded $r_{\max}$, as is typically the case for most CFGs, the loop over reachable states has time and space complexity of $O(|\Delta_{g}| \cdot |\Delta_{r}|)$.
  \end{itemize}

  \item \textbf{Finding and removing duplicate states} (Algorithm 3.4 \textsc{FindDupStates}):
  \begin{itemize}
    \item In Algorithm 3.4, for each pair of states, comparing their sets of transitions is bounded by running time of $O(|Q|^{2} \cdot |\Delta|)$, where $Q$ and $\Delta$ respectively refer to the numbers of states and of transitions after the main loop in Algorithm 3.3.

    \item At most $O(|Q|^{2})$ pairs of states exist, and storing the per-state transition sets has space cost of $O(|\Delta|)$, resulting in Algorithm 3.4 with a space complexity of $O(|Q|^{2} + |\Delta|)$.
  \end{itemize}

  (Back to Algorithm 3.3:)
  \begin{itemize}
    \item Then, running merge/removal operation on all pairs of duplicate state pairs ($Q_{dup}$) has time cost of $O(|Q|^{2})$ and space cost of $O(|Q|)$.

    \item Rewriting transitions by making pass over all transitions has time and space cost of $O(|\Delta|)$.

    \item So, this part---dominated by detection cost---has time and space complexities of $O(|Q|^{2} \cdot |\Delta|)$ and of $O(|Q|^{2} + |\Delta|)$.
  \end{itemize}

  \item \textbf{$\varepsilon$-transition introduction}: This part traverses the ordered states and for each pair, checks the corresponding transitions. This part is negligible compared to the duplicate detection process.

  \item Combining the above together, Algorithm 3.3 has the time complexity of $O((|Q_g| \cdot |Q_r|)^{2} \cdot |\Delta_g| \cdot |\Delta_r|)$ and space complexity of $O((|Q_g| \cdot |Q_r|)^{2} + |\Delta_g| \cdot |\Delta_r|)$.
\end{itemize}

\section{User Interaction: Menhir vs. \tool}
\label{sup:ui}

This section provides a concrete example illustrating the difference
between the user interaction offered by \menhir, showing the form and
complexity of the information presented to users during manual
disambiguation, and to contrast it with the example-based interaction
used by \tool when resolving grammar ambiguities. 
The example is drawn from our experience repairing the Michelson
grammar $G7$ from the paper benchmark.
This section also indicates that manual disambiguation effort is
dominated not only by the number of ambiguities, but also by the
cognitive cost of interpreting automaton-level conflict reports and
anticipating the effects of grammar edits. 

\paragraph{\menhir conflict report.}
When applied to grammar $G7$, \menhir generates a conflict report
in~\autoref{fig:menhir-report}. 
We include the report largely verbatim to demonstrate the level of details exposed to the user. 
Resolving this conflict manually requires the user to inspect the
automaton state, identify the relevant productions, reconstruct the
competing parse trees, and infer which grammar modification reflects
the intended parsing preference. 
This process demands both familiarity with the grammar and experience
with LR parser diagnostics. 

\begin{figure}[t]
\centering
\begin{minipage}{0.95\linewidth}
\begin{lstlisting}[style=HL,numbers=none]
** Conflict (shi(*@f@*)t/reduce) in state 134.
** Token involved: ELT
** This state is reached from toplevel after reading:

CODE LBRACE MNEMONIC tyy SOME literal

** The derivations that appear below have the following common factor:
** (The question mark symbol (?) represents the spot where the derivations begin to di(*@f@*)fer.)

toplevel
script EOF
CODE LBRACE instlist RBRACE
(*@\hspace{6em}@*)singleinst
(*@\hspace{6em}@*)MNEMONIC tyy literal
(*@\hspace{13em}@*)(?)

** In state 134, looking ahead at ELT, reducing production..

literal -(*@>@*) SOME literal

** is permitted because of the following sub-derivation:

literal ELT literal // lookahead token appears
SOME literal .

** In state 134, looking ahead at ELT, shi(*@f@*)ting is permitted..
\end{lstlisting}
\end{minipage}
\caption{Excerpt from a \menhir conflict report for the Michelson
grammar \texttt{G7}.}
\label{fig:menhir-report}
\end{figure}

\paragraph{\tool user interface.} 
For the same ambiguity, \tool presents the user with the prompt
in~\autoref{fig:greta-ui}, exposing the alternative parse structures
involved in the conflict. 
Here, the ambiguity is expressed explicitly as a choice between concrete parse alternatives. 
Then, the user resolves the conflict by selecting the parse that
matches their intent, without inspecting parser states or conflict
reports. 
This example-based interaction avoids accidental over-disambiguation
and backtracking, and leads to more predictable repair behaviour.

\begin{figure}[th]
\centering
\begin{minipage}{0.95\linewidth}
\begin{lstlisting}[style=HL,numbers=none]
Choose your preference(*@!@*)
(Type either 0 or 1.)

Option 0:
(*@\hspace{4em}@*)( SOME ( literal ELT literal ) )
Option 1:
(*@\hspace{4em}@*)( ( SOME literal ) ELT literal )
\end{lstlisting}
\end{minipage}
\caption{User interface presented by \tool for the same ambiguity shown in~\autoref{fig:menhir-report}.}
\label{fig:greta-ui}
\end{figure}

}{
  \section*{Data Availability}
\label{sec:data-availability}

The software artefact accompanying this paper is available
online~\cite{greta-artifact}.
It contains the OCaml implementation of \greta, including
\textsc{GenTA} (\autoref{subsec:learner}) and \textsc{IntersectTA}
(\autoref{subsec:intersect}), as well as benchmark data and build
scripts for reproducing the evaluation results reported
in~\autoref{sec:eval}. 

  \bibliography{references}
}

\end{document}
